\documentclass[12pt]{article}

\usepackage{color}
\usepackage{enumitem}
\usepackage{amssymb, amsmath}
\usepackage{hyperref}
\usepackage{tikz}% Pour dessiner un graphe

\usepackage[left=1.5cm,right=1.5cm,top=2cm,bottom=2.5cm]{geometry}

\usepackage{caption}
\newcommand\EFFACE[1]{}

\newtheorem{theorem}{Theorem}
\newtheorem{proposition}[theorem]{Proposition}

\newtheorem{lemma}[theorem]{Lemma}
\newtheorem{corollary}[theorem]{Corollary}

\newtheorem{claim}{Claim}
\newtheorem{observation}[theorem]{Observation}

\newenvironment{proof}{
\par
\noindent {\bf Proof.}\rm}{\mbox{}\hfill$\square$\par\vskip 3mm}

\newcommand\SOMMET[1]{\draw[fill=black] (#1) circle (2pt);}
\newcommand\ETIQUETTE[2]{\node[below] at (#1) {#2};}

\def\ch{{\rm ch}}
\def\dist{{\rm dist}}

\makeatletter
\let\@fnsymbol\@arabic
\makeatother

\begin{document}

%\vspace*{2cm}

%\title{{\bf On the List Incidence chromatic number of graphs}} 
\title{{\bf Incidence Choosability of Graphs}} 

\author{Brahim BENMEDJDOUB~\thanks{Faculty of Mathematics, Laboratory L'IFORCE, University of Sciences and Technology
Houari Boumediene (USTHB), B.P.~32 El-Alia, Bab-Ezzouar, 16111 Algiers, Algeria.}
\and Isma BOUCHEMAKH~\footnotemark[1]
\and \'Eric SOPENA~\thanks{Univ. Bordeaux, LaBRI, UMR5800, F-33400 Talence, France.}~$^,$\thanks{CNRS, LaBRI, UMR5800, F-33400 Talence, France.}~$^,$\footnote{Corresponding author. Eric.Sopena@labri.fr.}
}

\maketitle

\abstract{
An incidence of a graph $G$ is a pair $(v,e)$ where $v$ is a vertex of~$G$ and $e$ is an edge of~$G$ incident with $v$. Two incidences $(v,e)$ and $(w,f)$ of~$G$ are adjacent whenever
(i) $v=w$, or (ii) $ e=f$, or (iii) $vw=e$ or $f$.
An incidence $p$-colouring of~$G$ is a mapping from the set of incidences of~$G$
 to the set of colours $\{1,\dots,p\}$ such that every two adjacent incidences receive distinct colours.
Incidence colouring has been introduced by Brualdi and Quinn Massey in 1993 and, since then,
studied by several authors.

In this paper, we introduce and study the list version of incidence colouring.
We determine the exact value of -- or upper bounds on -- the incidence choice number of 
several classes of graphs, namely square grids, Halin graphs, %generalized coronae of cycles,
cactuses and Hamiltonian cubic graphs.
}%abstract

\medskip

\noindent
{\bf Keywords:} 
Incidence colouring; 
Incidence list colouring;
List colouring; 
Square grid; 
Halin graph; 
%Corona of a cycle; Cactus; 
Hamiltonian cubic graph.

\medskip

\noindent
{\bf MSC 2010:} 05C15.

%%%%%%%%%%%%%%%%%%%%%%%%%%%%%%%%%%%%%%%%%%%%%%%%%%%%%%%%%%%%%%%%%%%%%%%%%%%%%%%%%%%%%%%%%%%%%%%%%%%%%%%%%%%%%%%%%%%%%%%%%%%%%
%%%%%%%%%%%%%%%%%%%%%%%%%%%%%%%%%%%%%%%%%%%%%%%%%%%%%%%%%%%%%%%%%%%%%%%%%%%%%%%%%%%%%%%%%%%%%%%%%%%%%%%%%%%%%%%%%%%%%%%%%%%%%
%%%%%%%%%%%%%%%%%%%%%%%%%%%%%%%%% INTRODUCTION

\section{Introduction}

All graphs  considered in this paper are simple and loopless undirected graphs.
We denote by $V(G)$ and $E(G)$ the set of vertices and the set of edges of a graph $G$, respectively,
 by $\Delta(G)$ the maximum degree of~$G$,
 and by $\dist_G(u,v)$ the distance between vertices $u$ and $v$ in $G$.

A \emph{(proper) colouring} of a graph $G$ is a mapping from $V(G)$ to a finite set of colours
such that adjacent vertices are assigned distinct colours.
Let $L$ be a \emph{list assignment} of~$G$, that is, a mapping that assigns to every vertex $v$ of~$G$ a finite list $L(v)$ of colours.
The graph $G$ is \emph{$L$-list colourable} if there exists a proper colouring $\lambda$ of~$G$
satisfying $\lambda(v)\in L(v)$ for every vertex $v$ of~$G$.
The graph $G$ is \emph{$k$-list colourable}, or \emph{$k$-choosable}, if, 
for every list assignment $L$ with $|L(v)|=k$ for every vertex $v$,
$G$ is $L$-list colourable.
The \emph{choice number} $\ch(G)$ of~$G$ is then defined as the smallest integer $k$ such that $G$ is $k$-choosable.
List colouring was independently introduced by Vizing~\cite{V76} and Erd\H{o}s, Rubin and Taylor~\cite{ERT79}
(see the surveys by Alon~\cite{A93},  Tuza~\cite{T97},
Kratochv\`il, Tuza and Voigt~\cite{KTV99}, 
or the monography by Chartrand and Zhang~\cite[Section~9.2]{CZbook}).

An {\em incidence} of a graph $G$ is a pair $(v,e)$ where $v$ is a vertex of~$G$ and $e$ is an edge of~$G$ incident with $v$.
%We denote by $Inc(G)$ the set of incidences of~$G$.
Two incidences $(v,e)$ and $(w,f)$ of~$G$ are {\em adjacent} whenever
(i) $v=w$, or (ii) $ e=f$, or (iii) $vw=e$ or $f$.
An {\em incidence $p$-colouring} of~$G$ is a mapping from the set of incidences of~$G$
 to the set of colours $\{1,\dots,p\}$ such that every two adjacent incidences receive distinct colours.
The smallest $p$ for which $G$ admits an incidence $p$-colouring is the \emph{incidence chromatic number} of~$G$, denoted by $\chi_i(G)$.
Incidence colourings were first introduced and studied by Brualdi and Quinn Massey~\cite{BM93}.
Incidence colourings of various graph families have attracted much interest in recent years, 
see for instance \cite{GLS16a,GLS16b,HSZ04,M05,SW13,WCP02,W09}.

The list version of incidence colouring is defined in a way similar to the case of ordinary proper vertex colouring.
We thus say  that a graph $G$ is %\emph{$k$-list incidence colourable}, or 
\emph{incidence $k$-choosable}, if, 
for every list assignment $L$ with $|L(v,e)|=k$ for every incidence $(v,e)$,
$G$ is \emph{$L$-list incidence colourable}.
The \emph{incidence choice number} of~$G$, denoted by $\ch_i(G)$, is then defined as the smallest integer $k$ such that $G$ is 
incidence $k$-choosable.

Our paper is organised as follows.
We first give some preliminary results in Section~\ref{sec:preliminary}.
We then study the incidence choice number
of square grids in Section~\ref{sec:grids},
of Halin graphs in  Section~\ref{sec:Halin},
%of generalized coronae of cycles  in  Section~\ref{sec:coronae},
of cactuses in Section~\ref{sec:cactus},
and of Hamiltonian cubic graphs in Section~\ref{sec:cubic}.
We finally propose some directions for future research in Section~\ref{sec:discussion}.

%%%%%%%%%%%%%%%%%%%%%%%%%%%%%%%%%%%%%%%%%%%%%%%%%%%%%%%%%%%%%%%%%%%%%%%%%%%%%%%%%%%%%%%%%%%%%%%%%%%%%%%%%%%%%%%%%%%%%%%%%%%%%
%%%%%%%%%%%%%%%%%%%%%%%%%%%%%%%%%%%%%%%%%%%%%%%%%%%%%%%%%%%%%%%%%%%%%%%%%%%%%%%%%%%%%%%%%%%%%%%%%%%%%%%%%%%%%%%%%%%%%%%%%%%%%
%%%%%%%%%%%%%%%%%%%%%%%%%%%%%%%%% BASIC RESULTS

\section{Preliminary results}\label{sec:preliminary}

We list in this section some basic results on the incidence choice number
of various graph classes. 
Note first that the inequality $\ch_i(G)\ge\chi_i(G)$ %and $\chi_i(G)\ge\Delta(G)+1$ 
obviously holds for every graph $G$,
and that whenever $G$ is not connected, $\chi_i(G)$ (resp. $\ch_i(G)$)
equals the maximum value of $\chi_i(C)$ (resp. of $\ch_i(C)$), taken
over all connected components $C$ of~$G$.
Therefore, when studying the incidence chromatic number or the incidence choice number
of special graph classes, it is enough to consider the case of connected graphs.

We start by introducing some notation.
With any graph $G$, we associate the \emph{incidence graph of~$G$}, denoted by $I_G$,
whose vertices are the incidences of~$G$, two incidences being joined by an edge
whenever they are adjacent. Clearly, every incidence colouring of~$G$ is nothing
but a proper vertex colouring of $I_G$,
so that $\chi_i(G)=\chi(I_G)$ and $\ch_i(G)=\ch(I_G)$.
Note also that for every subgraph $H$ of $G$, $I_H$ is a subgraph of $I_G$.
Hence we have:

\begin{observation}
For every subgraph $H$ of a graph $G$, $\chi_i(H)\leq \chi_i(G)$ and $\ch_i(H)\leq \ch_i(G)$.
\label{obs:subgraph}
\end{observation}

For every vertex $v$ in a graph $G$, we denote by
$A^-(v)$ the set of incidences of the form $(v,vu)$,
and by $A^+(v)$ the set of incidences of the form $(u,uv)$
(see Figure~\ref{fig:adjacent-incidences}).
We thus have $|A^-(v)|=|A^+(v)|=\deg(v)$ for every vertex $v$.
For every vertex $v$, the incidences in $A^-(v)$ are called
the \emph{internal incidences of $v$}, and
the incidences in $A^+(v)$ are called
the \emph{external incidences of $v$}.
The following observation will be useful.

\begin{observation}
For every incidence $(v,vu)$, the set of incidences that are
adjacent to $(v,vu)$ is $A^-(v) \cup A^+(v) \cup A^-(u)$,
whose cardinality is $2\deg_G(v)+\deg_G(u)-2$.
\label{obs:incidences}
\end{observation}

Note also that all incidences in $A^-(v)$ must be assigned pairwise distinct
colours in every incidence colouring of~$G$ and that the colour of any
incidence in $A^+(v)$ must be distinct from the colours assigned to
the incidences of $A^-(v)$.
Moreover, since every incidence has 
 at most $3\Delta(G)-2$ adjacent incidences by Observation~\ref{obs:incidences}
(see Figure~\ref{fig:adjacent-incidences}), we get:

\begin{proposition}
For every graph $G$, $\Delta(G)+1 \leq \chi_{i}(G)\leq \ch_i(G)\leq 3\Delta(G)-1$.
\label{prop:bounds}
\end{proposition}

It was proved in~\cite{ERT79,V76} that the choice number also satisfies a Brooks-like theorem, that is,
the inequality $\ch(G)\le\Delta(G)$ holds for every graph $G$ which is neither complete nor
an odd cycle.
Observe that whenever $\Delta(G)\ge 2$, the incidence graph $I_G$ contains a triangle
(induced by three incidences of the form $(v,vu_1)$, $(v,vu_2)$
and $(u_1,u_1v)$, $u_1\neq u_2$) and is non-complete (two incidences of the form $(u_1,u_1v)$ 
and $(u_2,u_2v)$, $u_1\neq u_2$, are not adjacent).
On the other hand, if $\Delta(G)=1$, then $G$ is a union of $K_2$'s, and thus
incidence 2-colourable.
Hence, Proposition~\ref{prop:bounds} can be slightly improved as follows:

\begin{proposition}
For every graph $G$ with $\Delta(G)\ge 2$, 
$\Delta(G)+1 \leq \chi_{i}(G)\leq \ch_i(G)\leq 3\Delta(G)-2.$
\label{prop:bounds-2}
\end{proposition}

%%%%
Recall that for every integer $p\ge 1$, the {\it $p^{th}$-power} $G^p$ of a graph $G$
is the graph obtained from~$G$ by linking every two 
vertices at distance at most $p$ from each other in $G$, that is,
$V(G^p)=V(G)$ and $uv\in E(G^p)$ if and only if
$1\le \dist_G(u,v)\le p$.
Consider now the cycle $C_n$ of order $n\ge 3$.
Such a cycle has $2n$ incidences and the associated
incidence graph $I_{C_n}$ is the square $C_{2n}^2$ of the cycle $C_{2n}$.
In~\cite{PW03}, Prowse and Woodall proved that
$\ch(C_n^p)=\chi(C_n^p)$ for every $p\ge 1$ and $n\ge 3$,
and thus, in particular, for the square of such a cycle.
On the other hand, it is not difficult to determine the incidence chromatic number
of any cycle $C_n$~\cite{BM93,SopWeb}.
Therefore, we get:

\begin{theorem}
For every $n\ge 3$, $3\le\ch_i(C_n)= \chi_{i}(C_n)\le 4$,
with $\ch_i(C_n)= \chi_{i}(C_n)=3$ if and only if $n\equiv 0\pmod 3$.
\label{th:cycles}
\end{theorem}

%%%%%%%%%%%%%%%%%%%%%%%%%%%%%%%%%%%%%%%%%%%%%%%%%%%%%%%%%%%%%%%%%%%%
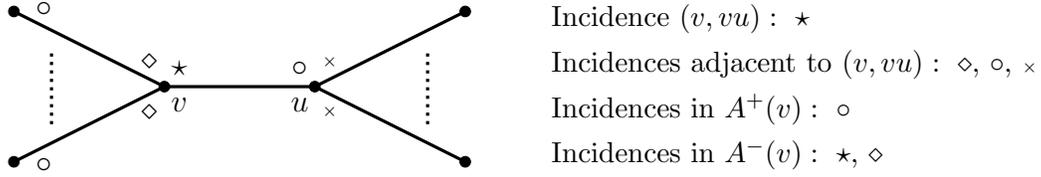
\begin{figure}
\centering
\begin{tikzpicture}[x=1cm,y=1cm]
%sommets
\SOMMET{0,0} \SOMMET{2,0}
\SOMMET{-2,-1} \SOMMET{-2,1} 
\SOMMET{4,-1} \SOMMET{4,1} 
\ETIQUETTE{0.2,0}{$v$} \ETIQUETTE{1.8,0}{$u$}
%aretes
\draw[very thick] (0,0) -- ++(2,0);
\draw[very thick] (0,0) -- ++(-2,-1);
\draw[very thick] (0,0) -- ++(-2,1);
\draw[very thick] (2,0) -- ++(2,-1);
\draw[very thick] (2,0) -- ++(2,1);
%pointillés
\draw[very thick,dotted] (-1.5,-0.5) -- ++(0,1);
\draw[very thick,dotted] (3.5,-0.5) -- ++(0,1);
%incidences
\node[above] at (0.2,0) {$\star$};
\node[above] at (-0.2,0.1) {$\diamond$};
\node[below] at (-0.2,-0.1) {$\diamond$};
\node[above] at (1.8,0) {$\circ$};
\node[above] at (-1.6,0.8) {$\circ$};
\node[below] at (-1.6,-0.8) {$\circ$};
\node[above] at (2.2,0.1) {{\tiny $\times$}};
\node[below] at (2.2,-0.1) {{\tiny $\times$}};
%légende
\node[right] at (5,0.9) {\small Incidence $(v,vu):\ \star$};
\node[right] at (5,0.3) {\small Incidences adjacent to $(v,vu):\ \diamond$, $\circ$, {\tiny $\times$}};
\node[right] at (5,-0.3) {\small Incidences in $A^+(v):\ \circ$};
\node[right] at (5,-0.9) {\small Incidences in $A^-(v):\ \star$, $\diamond$};
\end{tikzpicture}
\caption{Adjacent incidences.}
\label{fig:adjacent-incidences}
\end{figure}
%%%%%%%%%%%%%%%%%%%%%%%%%%%%%%%%%%%%%%%%%%%%%%%%%%%%%%%%%%%%%%%%%%%%%

A graph $G$ is \emph{$d$-degenerated} if every subgraph of~$G$ contains a vertex of degree at most $d$.
By a simple inductive argument, it is easy to prove that every $d$-degenerate graph has 
chromatic number, as well as choice number,
at most $d+1$~\cite[Proposition 2.2]{A93}.
Let $v$ be any vertex of~$G$ with degree at most $d$.
Every incidence of the form $(v,vu)$ has then at most $\Delta(G)+2d-2$ adjacent incidences in $G$.
Therefore, the incidence graph $I_G$ is $(\Delta(G)+2d-2)$-degenerate whenever $G$ is $d$-degenerate,
and we have:

\begin{theorem}
For every $d$-degenerated graph $G$, $\chi_i(G)\le\ch_i(G)\le\Delta(G)+2d-1$.
\label{th:degenerated}
\end{theorem}

%In~\cite{HSZ04}, Hosseini Dolama and Sopena proved that every $d$-degenerated graph $G$ admits an incidence
%colouring such that, for every vertex $v$ of~$G$,
%the set $A^+(v)$ uses at most $d$ distinct colours.

Since every tree is 1-degenerated, every $K_4$-minor free graph (and thus every outerplanar graph) is 2-degenerated,
%every square grid is 2-degenerated
and every planar graph is 5-degenerated, Theorem~\ref{th:degenerated} gives
the following:

\begin{corollary}For every graph $G$,
\begin{enumerate}
  \item if $G$ is a  tree, then  $\ch_i(G) = \Delta(G)+1$,
  \item if $G$ is a $K_4-$minor free graph (resp. an outerplanar graph), then  $\ch_i(G)\leq \Delta(G)+3$,
  \item if $G$ is a planar graph,  then $\ch_i(G)\leq \Delta(G)+9$.
\end{enumerate}
\label{cor:degenerated}
\end{corollary}

%%%%%%%%%%%%%%%%%%%%%%%%%%%%%%%%%%%%%%%%%%%%%%%%%%%%%%%%%%%%%%%%%%%%%%%%%%%%%%%%%%%%%%%%%%%%%%%%%%%%%%%%%%%%%%%%%%%%%%%%%%%%%
%%%%%%%%%%%%%%%%%%%%%%%%%%%%%%%%%%%%%%%%%%%%%%%%%%%%%%%%%%%%%%%%%%%%%%%%%%%%%%%%%%%%%%%%%%%%%%%%%%%%%%%%%%%%%%%%%%%%%%%%%%%%%
%%%%%%%%%%%%%%%%%%%%%%%%%%%%%%%%% SQUARE GRIDS

\section{Square grids}\label{sec:grids}

The \emph{square grid} $G_{m,n}$ is the graph defined as the Cartesian product of two paths of respective order $m$ and $n$, 
that is, $G_{m,n}=P_m\,\Box\, P_n$.
Since every square grid is 2-degenerated, 
Theorem~\ref{th:degenerated} gives $\ch_i(G_{m,n})\le\Delta(G_{m,n})+3\le 7$ for every $m$ and $n$, $m\ge n\ge 2$.
In this section, we prove that this bound can be decreased to~5 if $n=2$ and to~6 if $n\ge 3$.

We first prove the following useful lemma.

%%%%%%%%%%%%%%%%%%%%%%%%%%%%%%%%%%%%%%%%%%%%%%%%%%%%%%%%%%%%%%
\begin{figure}
\begin{center}
\begin{tikzpicture}[domain=0:10,x=0.7cm,y=0.7cm]
%sommets
\SOMMET{-3,0} \SOMMET{0,0} \SOMMET{3,0} \SOMMET{6,0}
\SOMMET{0,-3} \SOMMET{3,-3} 
\SOMMET{0,3} \SOMMET{3,3}  \SOMMET{6,3}
\SOMMET{3,6} 
%arêtes
\draw[very thick] (-3,0) -- (6,0);
\draw[very thick] (0,3) -- (6,3);
\draw[very thick] (0,-3) -- (0,3);
\draw[very thick] (3,-3) -- (3,6);
%noms sommets
\node[above] at (0.3,0) {$u$};
\node[above] at (2.7,0) {$x$};
\node[below] at (0.3,3) {$v$};
\node[below] at (2.7,3) {$w$};
\node[left] at (-3,0.1) {$u'$};
\node[below] at (0,-3) {$u''$};
\node[below] at (3,-3) {$x''$};
\node[right] at (6,0.1) {$x'$};
\node[right] at (6,3.1) {$w''$};
\node[above] at (3,6) {$w'$};
%labels
\node[above] at (-2.5,-0.8) {$\alpha'_1$};
\node[above] at (-0.5,-0.8) {$\alpha_1$};
\node[left] at (0,0.7) {$\alpha_2$};
\node[left] at (0,2.3) {$\alpha'_2$};
\node[above] at (2.3,3) {$\beta_1$};
\node[above] at (3.7,3) {$\beta_3$};
\node[right] at (3,2.3) {$\beta_4$};
\node[right] at (3,4.2) {$\beta_2$};
\node[below] at (0.8,0) {$[a]$};
\node[below] at (2.2,0) {$[c]$};
\node[right] at (3,0.7) {$[d]$};
\node[right] at (0,-1.2) {$[b]$};
\end{tikzpicture}
\caption{The partially $L_0$-list incidence coloured graph $H_0$ of Lemma~\ref{lem:grid}.}
\label{fig:grid}
\end{center}
\end{figure}
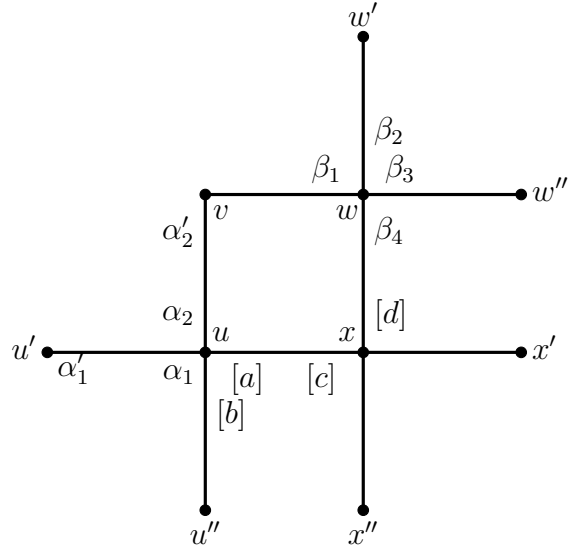
%%%%%%%%%%%%%%%%%%%%%%%%%%%%%%%%%%%%%%%%%%%%%%%%%%%%%%%%%%%%%%%%%%%%%%%%

\begin{lemma}
Let $H_0$ be the graph depicted in Figure~\ref{fig:grid}, $L_0$ be any list assignment of $H_0$
such that $|L_0(i,ij)|\ge 6$ for every incidence $(i,ij)$ of $H_0$,
and $\sigma_0$ be the partial
$L_0$-list incidence colouring of $H_0$ using colours from the set
$\{\alpha_1,\alpha'_1,\alpha_2,\alpha'_2\,\beta_1,\beta_2,\beta_3,\beta_4\}$ 
depicted in Figure~\ref{fig:grid}.
Then, there exist 
$a\in L_0(u,ux)\setminus\{\alpha_1,\alpha'_1,\alpha_2,\alpha'_2\}$,
$b\in L_0(u,uu'')\setminus\{\alpha_1,\alpha'_1,\alpha_2,\alpha'_2\}$,
$c\in L_0(x,xu)\setminus\{\alpha_1,\alpha_2,\beta_4\}$, and
$d\in L_0(x,xw)\setminus\{\beta_1,\beta_2,\beta_3,\beta_4\}$,
such that $|\{a,b,c\}|=|\{a,c,d\}|=3$,
so that $\sigma_0$ can be extended to colour the four incidences
$(u,ux)$, $(u,uu'')$, $(x,xu)$ and $(x,xw)$.
\label{lem:grid}
\end{lemma}

\begin{proof}
Note first that $|L_0(u,ux)\setminus\{\alpha_1,\alpha'_1,\alpha_2,\alpha'_2\}|\ge 2$
and $|L_0(u,uu'')\setminus\{\alpha_1,\alpha'_1,\alpha_2,\alpha'_2\}|\ge 2$, so that
we can always choose $a$ and $b$ as required.

If $|L_0(x,xu)\cap\{\alpha_1,\alpha_2,\beta_4\}|\le 2$,
then we can choose 
$d\in L_0(x,xw)\setminus\{\beta_1,\beta_2,\beta_3,\beta_4,a\}$
and $c\in L_0(x,xu)\setminus\{\alpha_1,\alpha_2,\beta_4,a,b,d\}$.

Similarly, if $|L_0(x,xw)\cap\{\beta_1,\beta_2,\beta_3,\beta_4\}|\le 3$,
then we can choose
$c\in L_0(x,xu)\setminus\{\alpha_1,\alpha_2,\beta_4,a,b\}$
and $d\in L_0(x,xw)\setminus\{\beta_1,\beta_2,\beta_3,\beta_4,a,c\}$.

Suppose now that $\{\alpha_1,\alpha_2,\beta_4\}\subseteq L_0(x,xu)$
and $\{\beta_1,\beta_2,\beta_3,\beta_4\}\subseteq L_0(x,xw)$.
We consider three cases.
\begin{enumerate}
\item
If $\beta_4\in L_0(u,ux)\setminus\{\alpha_1,\alpha'_1,\alpha_2,\alpha'_2\}$,
then we set $a=\beta_4$.
We can then choose
$b\in L_0(u,uu'')\setminus\{\alpha_1,\alpha'_1,\alpha_2,\alpha'_2,\beta_4\}$,
$c\in L_0(x,xu)\setminus\{\alpha_1,\alpha_2,\beta_4,b\}$, and
$d\in L_0(x,xw)\setminus\{\beta_1,\beta_2,\beta_3,\beta_4,c\}$.

\item 
If $\beta_4\notin L_0(u,ux)\setminus\{\alpha_1,\alpha'_1,\alpha_2,\alpha'_2\}$
and $\beta_4\in L_0(u,uu'')\setminus\{\alpha_1,\alpha'_1,\alpha_2,\alpha'_2\}$,
then we set $b=\beta_4$.
Again, we can then choose
$a\in L_0(u,ux)\setminus\{\alpha_1,\alpha'_1,\alpha_2,\alpha'_2,\beta_4\}$,
$d\in L_0(x,xw)\setminus\{\beta_1,\beta_2,\beta_3,\beta_4,a\}$, and
$c\in L_0(x,xu)\setminus\{\alpha_1,\alpha_2,\beta_4,a,d\}$.

\item 
Suppose that none of the previous cases occurs.
Let $\{\varepsilon_1,\varepsilon_2\}\subseteq L_0(u,ux)\setminus\{\alpha_1,\alpha'_1,\alpha_2,\alpha'_2\}$ and
$\{\varepsilon_3,\varepsilon_4\}\subseteq L_0(u,uu'')\setminus\{\alpha_1,\alpha'_1,\alpha_2,\alpha'_2\}$.
We consider two subcases.
\begin{enumerate}
\item 
If $\{\varepsilon_1,\varepsilon_2\} \cap \{\varepsilon_3,\varepsilon_4\} = \emptyset$, 
we first choose
$d\in L_0(x,xw)\setminus\{\beta_1,\beta_2,\beta_3,\beta_4\}$ and
$c\in L_0(x,xu)\setminus\{\alpha_1,\alpha_2,\beta_4\}$
in such a way that $c\neq d$ and $\{\varepsilon_1,\varepsilon_2\}\neq\{c,d\}$
(this can be done since we have at least two choices for $d$, and then still two choices for $c$).
We then choose $a\in \{\varepsilon_1,\varepsilon_2\}\setminus\{c,d\}$
and $b\in \{\varepsilon_3,\varepsilon_4\}\setminus\{c\}$.

\item 
Otherwise, let $\mu\in \{\varepsilon_1,\varepsilon_2\} \cap \{\varepsilon_3,\varepsilon_4\}$.
We consider two subcases.
\begin{enumerate}
\item 
If $\mu\notin L_0(x,xu)$ or $\mu\notin L_0(x,xw)$, then we set $a=\mu$ and
$b\in \{\varepsilon_3,\varepsilon_4\}$ with $b\neq\mu$.
Now, if $\mu\notin L_0(x,xu)$,
we then choose 
$d\in L_0(x,xw)\setminus\{\beta_1,\beta_2,\beta_3,\beta_4,\mu\}$ and
$c\in L_0(x,xu)\setminus\{\alpha_1,\alpha_2,\beta_4,b,d\}$.
Otherwise (in which case we have $\mu\notin L_0(x,xw)$), 
we then choose 
$c\in L_0(x,xu)\setminus\{\alpha_1,\alpha_2,\beta_4,a,b\}$, and
$d\in L_0(x,xw)\setminus\{\beta_1,\beta_2,\beta_3,\beta_4,c\}$.

\item 
Suppose finally that $\mu\in L_0(x,xu) \cap L_0(x,xw)$.
If $\mu\notin\{\beta_1,\beta_2,\beta_3,\beta_4\}$, 
then we set $b=d=\mu$ and $a\in \{\varepsilon_1,\varepsilon_2\}$ with $a\neq\mu$.
We then choose
$c\in L_0(x,xu)\setminus\{\alpha_1,\alpha_2,\beta_4,a,\mu\}$.
Otherwise (that is, $\mu\in\{\beta_1,\beta_2,\beta_3,\beta_4\}$),
we set $a=\mu$ and
$b\in \{\varepsilon_3,\varepsilon_4\}$ with $b\neq\mu$,
so that we can choose
$c\in L_0(x,xu)\setminus\{\alpha_1,\alpha_2,\beta_4,\mu,b\}$ and
$d\in L_0(x,xw)\setminus\{\beta_1,\beta_2,\beta_3,\beta_4,c\}$.

\end{enumerate}

\end{enumerate}

\end{enumerate}

In all cases, the colours $a$, $b$, $c$ and $d$ clearly satisfy the requirements of the lemma.
\end{proof}

%%%%%%%%%%%%%%%%%%%%%%%%%%%%%%%%%%%%%%%%%%%%%%%%%%%%%%%%%%%%%%
\begin{figure}
\begin{center}
\begin{tikzpicture}[domain=0:10,x=0.7cm,y=0.7cm]
%sommets
\SOMMET{0,-3} \SOMMET{3,-3} \SOMMET{6,-3} \SOMMET{9,-3} 
\SOMMET{0,0} \SOMMET{3,0} \SOMMET{6,0} \SOMMET{9,0} 
\SOMMET{0,3} \SOMMET{3,3} \SOMMET{6,3} \SOMMET{9,3} 
\SOMMET{0,6} \SOMMET{3,6} \SOMMET{6,6} \SOMMET{9,6} 
\SOMMET{0,9} \SOMMET{3,9} \SOMMET{6,9} \SOMMET{9,9} 

%arêtes
\draw[very thick] (0,-3) -- (9,-3);
\draw[very thick] (0,0) -- (9,0);
\draw[very thick] (0,3) -- (9,3);
\draw[very thick] (0,6) -- (9,6);
\draw[very thick] (0,9) -- (9,9);
\draw[very thick] (0,-3) -- (0,9);
\draw[very thick] (3,-3) -- (3,9);
\draw[very thick] (6,-3) -- (6,9);
\draw[very thick] (9,-3) -- (9,9);

%step 1
\node[above] at (0.4,8.9) {{\scriptsize 1}};
\node[above] at (2.6,8.9) {{\scriptsize 1}};
\node[above] at (3.4,8.9) {{\scriptsize 1}};
\node[above] at (5.6,8.9) {{\scriptsize 1}};
\node[above] at (6.4,8.9) {{\scriptsize 1}};
\node[above] at (8.6,8.9) {{\scriptsize 1}};

\node[left] at (0,8.6) {{\scriptsize 1}};
\node[left] at (3,8.6) {{\scriptsize 1}};
\node[left] at (6,8.6) {{\scriptsize 1}};
\node[left] at (9,8.6) {{\scriptsize 1}};

\node[left] at (0,6.7) {{\scriptsize 1}};
\node[left] at (0,5.6) {{\scriptsize 1}};
\node[left] at (0,3.7) {{\scriptsize 1}};
\node[left] at (0,2.6) {{\scriptsize 1}};
\node[left] at (0,0.7) {{\scriptsize 1}};
\node[left] at (0,-0.4) {{\scriptsize 1}};
\node[left] at (0,-2.3) {{\scriptsize 1}};

\node[above] at (0.4,6) {{\scriptsize 1}};
\node[above] at (0.4,3) {{\scriptsize 1}};
\node[above] at (0.4,0) {{\scriptsize 1}};
\node[above] at (0.4,-3) {{\scriptsize 1}};

%step 2
\node[above] at (2.6,5.9) {{\scriptsize 2}};
\node[above] at (3.4,5.9) {{\scriptsize 2}};
\node[above] at (5.6,5.9) {{\scriptsize 2}};
\node[above] at (6.4,5.9) {{\scriptsize 2}};
\node[above] at (8.6,5.9) {{\scriptsize 2}};

\node[left] at (3,6.7) {{\scriptsize 2}};
\node[left] at (6,6.7) {{\scriptsize 2}};
\node[left] at (9,6.7) {{\scriptsize 2}};

\node[left] at (3,5.6) {{\scriptsize 2}};
\node[left] at (6,5.6) {{\scriptsize 2}};
\node[left] at (9,5.6) {{\scriptsize 2}};

%step 3
\node[above] at (2.6,2.9) {{\scriptsize 3}};
\node[above] at (3.4,2.9) {{\scriptsize 3}};
\node[above] at (5.6,2.9) {{\scriptsize 3}};
\node[above] at (6.4,2.9) {{\scriptsize 3}};

\node[left] at (3,3.7) {{\scriptsize 3}};
\node[left] at (6,3.7) {{\scriptsize 3}};

\node[left] at (3,2.6) {{\scriptsize 3}};
\node[left] at (6,2.6) {{\scriptsize 3}};

\node[above] at (2.6,-0.1) {{\scriptsize 3}};
\node[above] at (3.4,-0.1) {{\scriptsize 3}};
\node[above] at (5.6,-0.1) {{\scriptsize 3}};
\node[above] at (6.4,-0.1) {{\scriptsize 3}};

\node[left] at (3,0.7) {{\scriptsize 3}};
\node[left] at (6,0.7) {{\scriptsize 3}};

\node[left] at (3,-0.4) {{\scriptsize 3}};
\node[left] at (6,-0.4) {{\scriptsize 3}};

%step 4
\node[above] at (8.6,2.9) {{\scriptsize 4}};
\node[left] at (9,3.7) {{\scriptsize 4}};
\node[left] at (9,2.6) {{\scriptsize 4}};

\node[above] at (8.6,-0.1) {{\scriptsize 4}};
\node[left] at (9,0.7) {{\scriptsize 4}};
\node[left] at (9,-0.4) {{\scriptsize 4}};

%step 5
\node[above] at (2.6,-3.1) {{\scriptsize 5}};
\node[above] at (3.4,-3.1) {{\scriptsize 5}};
\node[above] at (5.6,-3.1) {{\scriptsize 5}};
\node[above] at (6.4,-3.1) {{\scriptsize 5}};
\node[above] at (8.6,-3.1) {{\scriptsize 5}};

\node[left] at (3,-2.3) {{\scriptsize 5}};
\node[left] at (6,-2.3) {{\scriptsize 5}};
\node[left] at (9,-2.3) {{\scriptsize 5}};

\end{tikzpicture}
\caption{Colouring the incidences of $G_{5,4}$ in five steps (Theorem~\ref{th:grid}).}
\label{fig:grid-algo}
\end{center}
\end{figure}
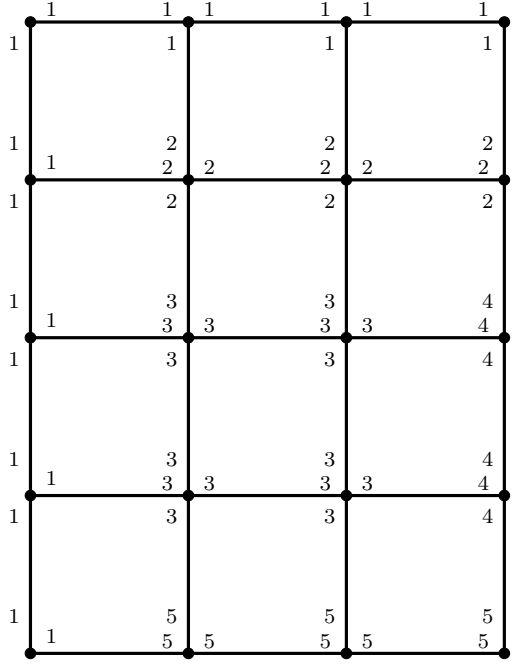
%%%%%%%%%%%%%%%%%%%%%%%%%%%%%%%%%%%%%%%%%%%%%%%%%%%%%%%%%%%%%%%%%%%%%%%%

We are now able to prove the main result of this section.

\begin{theorem}
For every integers $m$ and $n$, $m\ge n\ge 2$, we have
$$\left\{
  \begin{array}{ll}
     \ch_i(G_{m,n})\le 5, & \hbox{if $n=2$,} \\
     \ch_i(G_{m,n})\le 6, & \hbox{otherwise.}
  \end{array}
\right.$$
\label{th:grid}
\end{theorem}

\begin{proof}
Let $V(G_{m,n})=\{v_{i,j}\ |\ 1\le i\le m,\ 1\le j\le n\}$, so that
$E(G_{m,n})=\{(v_{i,j},v_{i',j'})\ |\ |i-i'|+|j-j'|=1\}$.

Suppose first that $n=2$ and let $L$ be any list assignment of $G_{m,2}$ such that $|L(v,vu)|= 5$
for every incidence $(v,vu)$ of $G_{m,2}$.
We construct an $L$-list incidence colouring of $G_{m,2}$ as follows.

Let us denote by $S_i$, $1\le i\le m-1$, the $i^{th}$ \emph{square} of $G_{m,2}$, that is,
the subgraph  of $G_{m,2}$ induced by the set of vertices 
$\{v_{i,1},v_{i,2},v_{i+1,1},v_{i+1,2}\}$.
We first colour the incidences of $S_1$.
This can be done since every such incidence has four adjacent incidences.

Then, if $m\ge 3$, we colour the incidences of the remaining squares
sequentially, from $S_2$ to $S_{m-1}$.
For each such square $S_i$, 
we colour the incidences
$(v_{i,1},v_{i,1}v_{i+1,1})$, $(v_{i+1,1},v_{i+1,1}v_{i,1})$, 
$(v_{i,2},v_{i,2}v_{i+1,2})$, $(v_{i+1,2},v_{i+1,2}v_{i,2})$, 
$(v_{i+1,1},v_{i+1,1}v_{i+1,2})$ and $(v_{i+1,2},v_{i+1,2}v_{i+1,1})$,
in that order.
This can be done since, doing so, every such incidence has 
at most four already coloured adjacent incidences.

\medskip

Suppose now that $m\ge n\ge 3$ and let $L$ be any list assignment of $G_{m,n}$ such that $|L(v,vu)|=6$
for every incidence $(v,vu)$ of $G_{m,n}$.
We will construct an $L$-list incidence colouring of $G_{m,n}$
in five steps.
Figure~\ref{fig:grid-algo} depicts the grid $G_{5,4}$ and gives, for each of its
incidences, the number (from $1$ to $5$) of the step during which it will be coloured.

\begin{enumerate}
\item 
We first colour all internal incidences of vertices $v_{1,j}$, sequentially 
from $v_{1,1}$ to $v_{1,n}$,
and all internal incidences of vertices $v_{i,1}$, sequentially, 
from $v_{2,1}$ to $v_{m,1}$. 
This can be done since, doing so, every such incidence has at most 
three already coloured adjacent incidences.

\item 
We then colour all internal incidences of vertices $v_{2,j}$, sequentially
from $v_{2,2}$ to $v_{2,n}$.
For each such vertex $v_{2,j}$, we colour its internal incidences
$(v_{2,j},v_{2,j}v_{2,j-1})$, $(v_{2,j},v_{2,j}v_{1,j})$, 
$(v_{2,j},v_{2,j}v_{3,j})$ and $(v_{2,j},v_{2,j}v_{2,j+1})$, in that order
(note that $v_{2,n}$ has only the first three internal incidences).
This can be done since, doing so, every
such incidence has at most five already coloured adjacent incidences.

\item 
Now, if $m\ge 4$, then, for $i=2$ to $m-1$, we colour the uncoloured internal incidences
of $v_{i,j}$, sequentially from $v_{i,2}$ to $v_{i,n-1}$ (when $n\ge 4$).
Each ``row'' of internal incidences, corresponding to vertices $v_{i,2}$ to $v_{i,n-1}$,
is coloured as follows.
\begin{enumerate}
\item 
We colour the internal incidences 
$(v_{i,2},v_{i,2}v_{i-1,2})$ and $(v_{i,2},v_{i,2}v_{i,1})$
of $v_{i,2}$, in that order, which can be done since these two incidences have 
five already coloured adjacent incidences.

\item 
If $2\le j\le n-2$, then the set of vertices
$$\{v_{i,j-1},v_{i,j},v_{i,j+1},v_{i,j+2},v_{i+1,j},v_{i+1,j+1},v_{i-1,j},v_{i-1,j+1},v_{i-1,j+2},v_{i-2,j+1}\}$$ 
induces a subgraph of $G_{m,n}$ isomorphic to the graph $H_0$ of Lemma~\ref{lem:grid}.
Therefore, according to Lemma~\ref{lem:grid}, the four incidences
$(v_{i,j},v_{i,j}v_{i,j+1})$, $(v_{i,j},v_{i,j}v_{i+1,j})$,
$(v_{i,j+1},v_{i,j+1}v_{i,j})$ and $(v_{i,j+1},v_{i,j+1}v_{i-1,j+1})$
can be coloured with the colours $a$, $b$, $c$ and $d$ given by the lemma,
respectively.

\item 
We finally colour the two incidences 
$(v_{i,n-1},v_{i,n-1}v_{i,n})$ and $(v_{i,n-1},v_{i,n-1}v_{i+1,n-1})$, in that order,
which can be done since, doing so, these incidences have four and five  already
coloured adjacent incidences, respectively.
\end{enumerate}

\item 
If $m\ge 4$, we colour all internal incidences of vertices $v_{i,n}$, $3\le i\le m-1$, sequentially from 
$v_{3,n}$ to $v_{m-1,n}$.
For each such vertex $v_{i,n}$, we colour its internal incidences
$(v_{i,n},v_{i,n}v_{i,n-1})$, $(v_{i,n},v_{i,n}v_{i-1,n})$ and 
$(v_{i,n},v_{i,n}v_{i+1,n})$, in that order.
This can be done since, doing so, every such incidence has at most five already coloured adjacent incidences.

\item 
Finally, we colour all (uncoloured) internal incidences of vertices $v_{m,j}$, sequentially from 
$v_{m,2}$ to $v_{m,n}$.
For each such vertex $v_{m,j}$, we colour its internal incidences
$(v_{m,j},v_{m,j}v_{m-1,j})$, $(v_{m,j},v_{m,j}v_{m,j-1})$ and $(v_{m,j},v_{m,j}v_{m,j+1})$,
in that order (note that $v_{m,n}$ has only the first two internal incidences).
This can be done since every such incidence has at most five already coloured adjacent incidences.
\end{enumerate}

This completes the proof.
\end{proof}

%%%%%%%%%%%%%%%%%%%%%%%%%%%%%%%%%%%%%%%%%%%%%%%%%%%%%%%%%%%%%%%%%%%%%%%%%%%%%%%%%%%%%%%%%%%%%%%%%%%%%%%%%%%%%%%%%%%%%%%%%%%%%
%%%%%%%%%%%%%%%%%%%%%%%%%%%%%%%%%%%%%%%%%%%%%%%%%%%%%%%%%%%%%%%%%%%%%%%%%%%%%%%%%%%%%%%%%%%%%%%%%%%%%%%%%%%%%%%%%%%%%%%%%%%%%
%%%%%%%%%%%%%%%%%%%%%%%%%%%%%%%%% HALIN GRAPHS

\section{Halin graphs}\label{sec:Halin}

Recall first that the \emph{star} $S_n$, $n\ge 1$, is the complete bipartite graph $K_{1,n}$.
%while the \emph{double-star} $D_{m,n}$, $m\ge n\ge 1$, is the graph obtained from the stars $S_m$ and $S_n$
%by adding an edge joining the central vertex of $S_m$ to the central vertex of $S_n$.
Moreover, the \emph{wheel} $W_n$, $n\ge 3$, is the graph obtained from the cycle $C_n$ by adding
a new vertex adjacent to every vertex of~$C_n$.

A \emph{Halin graph} is a planar graph obtained from a tree of order at least 4 with no vertex of degree~2, by adding a cycle connecting all its leaves~\cite{H71}.
We call this cycle the \emph{outer cycle of~$G$}.
In particular, every wheel is a Halin graph.
Wang, Chen and Pang proved that $\chi_i(G)=\Delta(G)+1$ for every Halin graph $G$ with $\Delta(G)\ge 5$~\cite{WCP02},
Shiu and Sun~\cite{SS08} that $\chi_i(G)=5$ for every cubic Halin graph,
and Meng, Guo and Su that $\chi_i(G)\le\Delta(G)+2$ for every Halin graph $G$ with $\Delta(G)=4$~\cite{MGS12}.
%On the other hand, it is known, by a result of Maydanskiy~\cite{M05}, that if $G$ is subcubic, then $\chi_i(G)\le 5$.

In this section, we determine the incidence choice number of every Halin graph $G$ with $\Delta(G)\ge 6$
and provide upper bounds for Halin graphs with smaller maximum degree.
For every Halin graph $G$, we denote by $C_G$ the outer cycle of~$G$
and by $T_G$ the subgraph of~$G$ obtained by deleting all the edges of the outer cycle of~$G$.
The subgraph $T_G$ is thus a tree and, in particular, $T_G$ is a star if $G$ is a wheel.

We will prove four lemmas, from which the main result of this section will follow.
%
%The following proposition will be useful in Section~\ref{sec:Halin}.
We first prove a preliminary result, which says that for every tree $T$ and list-assignment $L$ of $T$
with  $|L(v,vu)|\ge\Delta(T)+k$ for every incidence $(v,vu)$ of $T$ and some integer $k\ge 1$,
one can pre-colour $k$ %pendent 
incidences of $T$ and extend this pre-colouring
to an $L$-list incidence colouring of $T$.

%%%%%%%%%%%%%%%%%%%%%%%%%%%%%%%%%%%%%%%%%%%%%%%%%%%%%%%%%%%%%%%%%%%%
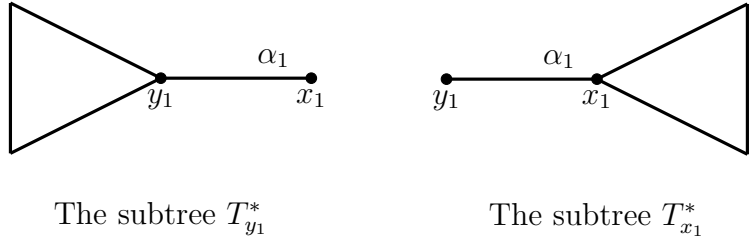
\begin{figure}
\centering
\begin{tikzpicture}[x=1cm,y=1cm]
%sommets
\SOMMET{0,0} \SOMMET{2,0}
%etiquettes
\node[below] at (0,0) {$y_1$};
\node[below] at (2,0) {$x_1$};
\node[above] at (1.5,0) {$\alpha_1$};
%arete
\draw[very thick] (0,0) -- (2,0);
%triangle
\draw[very thick] (0,0) -- (-2,1);
\draw[very thick] (0,0) -- (-2,-1);
\draw[very thick] (-2,1) -- (-2,-1);
%legende
\node[below] at (0,-1.5) {The subtree $T^*_{y_1}$};
\end{tikzpicture}
\hskip 1cm
\begin{tikzpicture}[x=1cm,y=1cm]
%sommets
\SOMMET{0,0} \SOMMET{2,0}
%etiquettes
\node[below] at (0,0) {$y_1$};
\node[below] at (2,0) {$x_1$};
\node[above] at (1.5,0) {$\alpha_1$};
%arete
\draw[very thick] (0,0) -- (2,0);
%triangle
\draw[very thick] (2,0) -- (4,1);
\draw[very thick] (2,0) -- (4,-1);
\draw[very thick] (4,1) -- (4,-1);
%legende
\node[below] at (2,-1.5) {The subtree $T^*_{x_1}$};
\end{tikzpicture}
\caption{Configurations for the proof of Proposition~\ref{prop:tree}.}
\label{fig:tree}
\end{figure}
%%%%%%%%%%%%%%%%%%%%%%%%%%%%%%%%%%%%%%%%%%%%%%%%%%%%%%%%%%%%%%%%%%%%%

\begin{proposition}
Let $T$ be a tree, $k\ge 1$ be an integer, and $L$ be a list-assignment of $T$ such that
$|L(v,vu)|\ge\Delta(T)+k$ for every incidence $(v,vu)$ in $T$.
For every set $\{(x_1,x_1y_1),\dots,(x_k,x_ky_k)\}$ of $k$ %pendent 
incidences in $T$
and every set $\{\alpha_1,\dots,\alpha_k\}$ of $k$ colours
such that $\alpha_i\in L(x_i,x_iy_i)$ for every~$i$, $1\le i\le k$,
and $\alpha_i\neq\alpha_j$ if $(x_i,x_iy_i)$ and $(x_j,x_jy_j)$ are adjacent, $1\le i<j\le k$,
there exists an $L$-list incidence colouring
$\sigma$ of $T$ such that $\sigma(x_i,x_iy_i)=\alpha_i$ for every~$i$, $1\le i\le k$.
\label{prop:tree}
\end{proposition}

\begin{proof}
The proof is by induction on $k$.
Let $L$ be a list-assignment of $T$ with $|L(v,vu)|\ge\Delta(T)+1$
for every incidence $(v,vu)$ in $T$, $(x_1,x_1y_1)$ be any incidence
%in $T$ (not necessarily pendent), 
in $T$, 
and $\alpha_1\in L(x_1,x_1y_1)$.
Let $T_{x_1}$ and $T_{y_1}$ denote the two components (trees) obtained
from $T$ by deleting the edge $x_1y_1$,
with $x_1\in V(T_{x_1})$ and $y_1\in V(T_{y_1})$.
We then denote by $T^*_{x_1}$ and $T^*_{y_1}$
the subtrees of $T$ obtained by adding the edge $x_1y_1$ to 
$T_{x_1}$ and $T_{y_1}$, respectively
(see Figure~\ref{fig:tree}),
and by $L_{x_1}$ and $L_{y_1}$ the restrictions of $L$
to $T^*_{x_1}$ and $T^*_{y_1}$, respectively.
The desired  $L$-list incidence colouring $\sigma$ of $T$ will be obtained
by combining 
an  $L_{x_1}$-list incidence colouring of $T^*_{x_1}$
with an $L_{y_1}$-list incidence colouring of $T^*_{y_1}$.

We construct $\sigma_{x_1}$ as follows. 
We first set $\sigma_{x_1}(x_1,x_1y_1)=\alpha_1$
and $\sigma_{x_1}(y_1,y_1x_1)=\beta_1$, for some $\beta_1\in L_{x_1}(y_1,y_1x_1)=L(y_1,y_1x_1)$.
Considering $y_1$ as the root of $T^*_{x_1}$, we can extend $\sigma_{x_1}$
to an  $L_{x_1}$-list incidence colouring of $T^*_{x_1}$ by colouring the incidences
in a top-bottom way, since, doing so, every uncoloured incidence will
have at most $\Delta(T^*_{x_1})\le\Delta(T)$ forbidden colours.
The colouring $\sigma_{y_1}$ is constructed similarly. 
We first set $\sigma_{y_1}(x_1,x_1y_1)=\alpha_1$
and $\sigma_{y_1}(y_1,y_1x_1)=\beta_1$,
and then colour the remaining incidences of $T^*_{y_1}$
in a top-bottom way, considering $x_1$ as the root of $T^*_{y_1}$.
Clearly, combining the colourings $\sigma_{x_1}$ and $\sigma_{y_1}$
produces an  $L$-list incidence colouring $\sigma$ of $T$ with $\sigma(x_1,x_1y_1)=\alpha_1$.

\smallskip

Suppose now that $k>1$.
Let $\{(x_1,x_1y_1),\dots,(x_k,x_ky_k)\}$ be a set of $k$ %pendent 
incidences in $T$
and $\{\alpha_1,\dots,\alpha_k\}$ be a set of $k$ colours
satisfying the conditions of the proposition.
Let $L'$ denote the list assignment of $T$ defined by
$L'(v,vu)=L(v,vu)\setminus\{\alpha_k\}$ for every incidence $(v,vu)$ in $T$.
Thanks to the induction hypothesis, there exists an  $L'$-list incidence colouring $\sigma'$
of $T$ such that $\sigma'(x_i,x_iy_i)=\alpha_i$ for every~$i$, $1\le i\le k-1$.
The required  $L$-list incidence colouring of $T$ is then obtained by
setting $\sigma(x_k,x_ky_k)=\alpha_k$ and $\sigma(v,vu)=\sigma'(v,vu)$
for every incidence $(v,vu)\neq(x_k,x_ky_k)$ in $T$.
\end{proof}

The next lemma gives a general upper bound
on the incidence choice number of Halin graphs.
Note that by Proposition~\ref{prop:bounds}, the corresponding bound is tight for every Halin graph
with maximum degree at least~6.

\begin{lemma}
If $G$ is a  Halin graph, then  $\ch_i(G)\leq \max(\Delta(G)+1, 7)$.
\label{lem:Halin1}
\end{lemma}

\begin{proof}
Let $G$ be a Halin graph and $L$ be any list assignment of~$G$ such that
$$|L(v,vu)|= \max(\Delta(G)+1, 7)\ge 7$$
 for every incidence $(v,vu)$ of~$G$.
Let $C_G=v_0v_1\dots v_{k-1}v_0$.
Each vertex~$v_i$, $0\le i\le k-1$, has thus three neighbours in $G$, namely $v_{i-1}$, $v_{i+1}$
(subscripts are taken modulo~$k$), and some vertex $t_i\in V(T_G)\setminus V(C_G)$
(see Figure~\ref{fig:Halin1}).
Note here that the $t_i$'s are not necessarily distinct.
%More precisely, we always have $t_i=t_{i-1}$ or $t_i=t_{i+1}$ (or both)
%for every $i$, $0\le i\le k-1$ (subscripts are taken modulo $k$).
Indeed, since every non-leaf vertex  of $T_G$ has degree at least~3,
we always have $t_{i-1}=t_i$ (when $v_{i-1}$ and $v_{i}$ have the same father in $T_G$),
 or $t_i=t_{i+1}$ (when $v_i$ and $v_{i+1}$ have the same father in $T_G$), and thus possibly
$t_{i-1}=t_i=t_{i+1}$,
for every $i$, $0\le i\le k-1$ (subscripts are taken modulo $k$).

By Corollary~\ref{cor:degenerated}, we know that
$T_G$ is incidence $(\Delta(T_G)+1)$-choosable, and thus incidence $(\Delta(G)+1)$-choosable.
Let $\sigma$ be such an  $L$-list incidence colouring of $T_G$.
Since every incidence of $C_G$ has exactly three already coloured adjacent
incidences in $T_G$, and thus at least four available colours in its list,
$\sigma$ can be extended to an
 $L$-list incidence colouring of~$G$, thanks to Theorem~\ref{th:cycles}.
\end{proof}

%%%%%%%%%%%%%%%%%%%%%%%%%%%%%%%%%%%%%%%%%%%%%%%%%%%%%%%%%%%%%%
\begin{figure}
\begin{center}
\begin{tikzpicture}[domain=0:10,x=0.8cm,y=0.8cm]
%sommets
\SOMMET{2,0} \SOMMET{5,0} \SOMMET{8,0} 
\SOMMET{2,2} \SOMMET{5,2} \SOMMET{8,2} 
%arêtes
\draw[very thick,dotted] (0,0) -- (1,0);
\draw[very thick] (1,0) -- (9,0);
\draw[very thick,dotted] (9,0) -- (10,0);
\draw[very thick] (2,0) -- (2,2);
\draw[very thick] (5,0) -- (5,2);
\draw[very thick] (8,0) -- (8,2);
%labels
\node[below] at (2,0) {$v_{k-1}$};
\node[below] at (5,0) {$v_{0}$};
\node[below] at (8,0) {$v_{1}$};
\node[above] at (2,2) {$t_{k-1}$};
\node[above] at (5,2) {$t_{0}$};
\node[above] at (8,2) {$t_{1}$};
\end{tikzpicture}
\caption{Part of the outer cycle $C_G$ of a Halin graph $G$ (the $t_i$'s are not necessarily distinct).}
\label{fig:Halin1}
\end{center}
\end{figure}
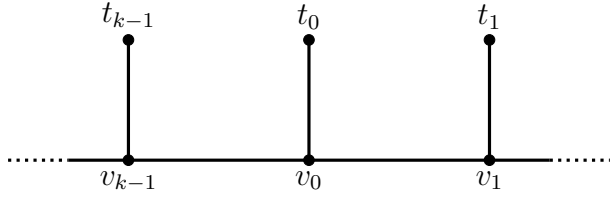
%%%%%%%%%%%%%%%%%%%%%%%%%%%%%%%%%%%%%%%%%%%%%%%%%%%%%%%%%%%%%%%%%%%%%%%%

Using Proposition~\ref{prop:tree}, we can get another upper bound on the
incidence choice number of Halin graphs that are not wheels.
This new bound thus improves the bound given in Lemma~\ref{lem:Halin1} for every Halin graph
with maximum degree~3 or~4, except for the two wheels $W_3=K_4$ and $W_4$.

\begin{lemma}
If $G$ is a  Halin graph such that $T_G$ is not a star, then $\ch_i(G)\leq \max(\Delta(G)+2, 6)$.
\label{lem:Halin2}
\end{lemma}

\begin{proof}
If $\Delta(G)\ge 5$, the result directly follows from Lemma~\ref{lem:Halin1}.
We can thus assume $\Delta(G)\in\{3,4\}$ (but we do not need this assumption in the proof).

Let $G$ be a Halin graph and $L$ be any list assignment of~$G$ such that
$$|L(v,vu)|= \max(\Delta(G)+2, 6)\ge 6$$
 for every incidence $(v,vu)$ of~$G$, and let $p=\max(\Delta(G)+2, 6)$.
As in the proof of Lemma~\ref{lem:Halin1}, we let $C_G=v_0v_1\dots v_{k-1}v_0$ and $t_i$ 
denotes the unique neighbour of $v_i$ in $V(T_G)\setminus V(C_G)$, $0\le i\le k-1$
(recall that the $t_i$'s are not necessarily distinct).
Note that starting from an $L$-list incidence colouring $\sigma$ of $T_G$ and then
colouring the incidences of $C_G$ in cyclic order, 
starting from any incidence,
all incidences of $C_G$ but the last two ones can be coloured, as each of these incidences has
at most five forbidden colours.
We will prove that one can always fix the colour of some incidences, so that
one can produce an $L$-list incidence colouring of~$G$.

Since $T_G$ is not a star, there exists an index $i$, $0\le i\le k-1$,
such that the vertices $t_{i-1}$ and $t_i$ are distinct.
We can thus assume, without loss of generality, that $t_{0}\neq t_1$.
Moreover, since $T_G$ has no vertex of degree two, we have
$t_{k-1}=t_{0}$ and $t_2=t_1$ (see Figure~\ref{fig:Halin2}).

%%%%%%%%%%%%%%%%%%%%%%%%%%%%%%%%%%%%%%%%%%%%%%%%%%%%%%%%%%%%%%%
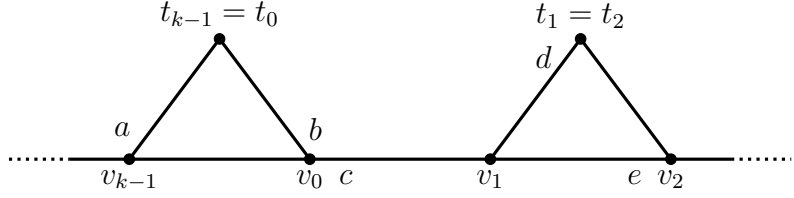
\begin{figure}
\begin{center}
\begin{tikzpicture}[domain=0:10,x=0.8cm,y=0.8cm]
%sommets
\SOMMET{2,0} \SOMMET{5,0} \SOMMET{8,0} \SOMMET{11,0}
\SOMMET{3.5,2} \SOMMET{9.5,2}
%arêtes
\draw[very thick,dotted] (0,0) -- (1,0);
\draw[very thick] (1,0) -- (12,0);
\draw[very thick,dotted] (12,0) -- (13,0);
\draw[very thick] (2,0) -- (3.5,2);
\draw[very thick] (5,0) -- (3.5,2);
\draw[very thick] (8,0) -- (9.5,2);
\draw[very thick] (11,0) -- (9.5,2);
%labels
\node[below] at (2,0) {$v_{k-1}$};
\node[below] at (5,0) {$v_{0}$};
\node[below] at (8,0) {$v_{1}$};
\node[below] at (11,0) {$v_{2}$};
\node[above] at (3.5,2) {$t_{k-1}=t_{0}$};
\node[above] at (9.5,2) {$t_{1}=t_2$};
% colours
\node[left] at (2.2,0.5) {$a$};
\node[below] at (5.6,0) {$c$};
\node[right] at (4.8,0.5) {$b$};
\node[left] at (9.2,1.7) {$d$};
\node[below] at (10.4,0) {$e$};
% symbol
%\node[below] at (5.6,0) {$\diamond$};
%\node[below] at (7.4,0) {$\star$};
\end{tikzpicture}
\caption{Configuration for the proof of Lemma~\ref{lem:Halin2}.}
\label{fig:Halin2}
\end{center}
\end{figure}
%%%%%%%%%%%%%%%%%%%%%%%%%%%%%%%%%%%%%%%%%%%%%%%%%%%%%%%%%%%%%%%%%%%%%%%%

The following claim will be essential in the construction of
an $L$-list incidence colouring of~$G$.

\begin{claim}\label{claim:Halin}
There exist $a\in L(v_{k-1},v_{k-1}t_{0})$, $b\in L(v_0,v_0t_0)$,
$c\in L(v_{0},v_{0}v_{1})$, $d\in L(t_1,t_1v_1)$ and $e\in L(v_2,v_2v_1)$,
with $b\neq c$,
such that 
$$|L(v_{k-1},v_{k-1}v_{0})\cap\{a,b,c\}|\le 2,\ 
|L(v_{0},v_0v_{k-1})\cap\{a,b,c\}|\le 2,
\ \mbox{and}\ |L(v_1,v_1v_{0})\cap\{c,d,e\}|\le 1.$$
\end{claim}

\begin{proof}
We first deal with the incidence $(v_1,v_1v_{0})$ and set the values of $c$, $d$ and $e$.
Let $C=L(v_{0},v_{0}v_{1})$, $D=L(t_1,t_1v_1)$ and $E=L(v_2,v_2v_1)$.
If $C\cap D\cap E\neq\emptyset$, then 
we set $c=d=e=\gamma$ for some $\gamma\in C\cap D\cap E$,
so that $|L(v_1,v_1v_{0})\cap\{c,d,e\}|\le 1$.

Otherwise, we have two cases to consider.
\begin{enumerate}
\item If $C$, $D$ and $E$ are pairwise disjoint, then at least two of them
are distinct from $L(v_1,v_1v_{0})$, so that we can choose $c$, $d$ and $e$
in such a way that $|L(v_1,v_1v_{0})\cap\{c,d,e\}|\le 1$.
\item 
Suppose now that $C\cap D\neq\emptyset$ 
(the cases $C\cap E\neq\emptyset$ and $D\cap E\neq\emptyset$ are similar).
We first set $c=d=\gamma$ for some $\gamma\in C\cap D$.
If $\gamma\in L(v_1,v_1v_{0})$, then there exists $\varepsilon\in E\setminus L(v_1,v_1v_{0})$
(since $(C\cap D)\cap E=\emptyset$) and we set $e=\varepsilon$, 
so that $|L(v_1,v_1v_{0})\cap\{c,d,e\}|\le 1$.
If $\gamma\notin L(v_1,v_1v_{0})$, then we set $e=\varepsilon$ for any $\varepsilon\in E$
and we also get $|L(v_1,v_1v_{0})\cap\{c,d,e\}|\le 1$.
\end{enumerate}

We now consider the incidence $(v_0,v_0v_{k-1})$.
Let $A=L(v_{k-1},v_{k-1}t_{0})$ and $B=L(v_0,v_0t_0)$.
If $c\notin L(v_0,v_0v_{k-1})$, then 
$|L(v_{0},v_0v_{k-1})\cap\{a,b,c\}|\le 2$ for any values of $a$ and $b$.

Suppose now that $c\in L(v_0,v_0v_{k-1})$.
If $|A\cap B|\ge 2$, then we set $a=b=\lambda$ for some $\lambda\in(A\cap B)\setminus\{c\}$,
so that $|L(v_{0},v_0v_{k-1})\cap\{a,b,c\}|\le 2$.
Otherwise, we necessarily have $A\neq L(v_{0},v_0v_{k-1})$ or $B\neq L(v_{0},v_0v_{k-1})$.
In the former case, we set $a=\alpha$ for some $\alpha\in A\setminus L(v_{0},v_0v_{k-1})$,
so that $|L(v_{0},v_0v_{k-1})\cap\{a,b,c\}|\le 2$ for any value of $b$.
In the latter case, we set $b=\beta$ for some $\beta\in B\setminus L(v_{0},v_0v_{k-1})$,
so that $|L(v_{0},v_0v_{k-1})\cap\{a,b,c\}|\le 2$ for any value of $a$.

\medskip

We finally consider the incidence $(v_{k-1},v_{k-1}v_0)$.
If $c\notin L(v_{k-1},v_{k-1}v_0)$, then 
$|L(v_{k-1},v_{k-1}v_{0})\cap\{a,b,c\}|\le 2$ for any values of $a$ and $b$ and we are done.

Suppose now that $c\in L(v_{k-1},v_{k-1}v_0)$.
If none of the values of $a$ and $b$ have been set in the previous step,
then we proceed as for the incidence $(v_0,v_0v_{k-1})$ and the result follows.
Otherwise, we have two cases to consider.
\begin{enumerate}
\item 
If the values of both $a$ and $b$ have been set in the previous step,
then we have $a=b=\lambda$, so that $|L(v_{k-1},v_{k-1}v_{0})\cap\{a,b,c\}|\le 2$.
\item 
Suppose now that the value of $a$ has been set in the previous step,
that is, $a=\alpha$ for some $\alpha\in A\setminus L(v_{0},v_0v_{k-1})$
(the proof is similar if the value of $b$ has been set).

If $\alpha\in B$, then we can set $b=\alpha$ and we are done.
If $\alpha\notin B$ and $\alpha\notin L(v_{k-1},v_{k-1}v_0)$, 
then we get $|L(v_{k-1},v_{k-1}v_{0})\cap\{a,b,c\}|\le 2$ for any value of $b$.
Otherwise, we have $\alpha\notin B$ and $\alpha\in L(v_{k-1},v_{k-1}v_0)$,
which implies $B\neq L(v_{k-1},v_{k-1}v_0)$.
Therefore, we can set $b=\beta$ for some $\beta\in B\setminus L(v_{k-1},v_{k-1}v_0)$,
so that $|L(v_{k-1},v_{k-1}v_{0})\cap\{a,b,c\}|\le 2$.
\end{enumerate}

This concludes the proof of Claim~\ref{claim:Halin}.
\end{proof}

We now construct an $L$-list incidence colouring $\sigma$ of~$G$ in three steps.
 
\begin{enumerate}
\item 
We first set 
$\sigma(v_{k-1},v_{k-1}t_0)=a$, 
$\sigma(v_0,v_0t_0)=b$, 
$\sigma(v_0,v_0v_1)=c$, 
$\sigma(t_1,v_1t_1)=d$, 
and $\sigma(v_2,v_1v_2)=e$,
where $a$, $b$, $c$, $d$ and $e$ are the values determined in the proof
of Claim~\ref{claim:Halin}.

\item
Let $P=t_0u_1\dots u_\ell t_1$, or $P=t_0t_1$  if $t_0t_1\in E(G)$,
denote the unique path from $t_0$ to $t_1$ in $T_G$ (see Figure~\ref{fig:path-t1-t0}).
We colour all the incidences of $T_G$ as follows.
\begin{itemize}
\item We first colour all internal incidences of $t_0$, starting with the incidence $(t_0,t_0v_0)$,
and then the incidence $(t_0,t_0t_1)$ if $t_0t_1\in E(G)$.
This can be done since every such incidence has at most $\Delta(G)+1$ already
coloured adjacent incidences.
\item If $t_0t_1\notin E(G)$), then we colour the internal incidences of the vertices of $P$ sequentially, 
from $u_1$ to $u_\ell$.
Again, every such incidence has at most $\Delta(G)+1$ already
coloured adjacent incidences.
\item We colour the incidence $(t_1,t_1u_\ell)$ (or $(t_1,t_1t_0)$ if $t_0t_1\in E(G)$), which 
has at most $\Delta(G)+1$ already coloured adjacent incidences,
then the incidence $(v_1,v_1t_1)$, which has four already
coloured adjacent incidences, 
and then the incidence $(t_1,t_1v_2)$,
which has five already coloured adjacent incidences (recall that $p\ge 6$).
\item We colour the remaining uncoloured internal incidences of $t_1$, if any.
This can be done since every such incidence has at most $\Delta(G)+1$ already
coloured adjacent incidences.
\item Now, we colour the uncoloured external incidences of the vertices of $P$,
sequentially, from $t_0$ to $t_1$.
Again, this can be done since every such incidence has at most $\Delta(G)+1$ already
coloured adjacent incidences.
\item 
For every edge $xy\in E(T_G)$, we denote by $T_{xy}$ the unique maximal subtree
of $T_G$ containing the edge $xy$ and such that $\deg_{T_{xy}}(x)=1$.
Clearly, each remaining uncoloured incidence of $T_G$ belongs to some subtree
$T_{xy}$, with $x\in V(P)$ and $y\notin V(P)\cup\{v_{k-1},v_0,v_1,v_2\}$.
Moreover, the only already coloured incidences of any such subtree $T_{xy}$
are $(x,xy)$ and $(y,yx)$.
By Proposition~\ref{prop:tree}, we can therefore extend $\sigma$ to 
all incidences of $T_G$.
\end{itemize}

\item We finally colour all the uncoloured incidences of $C_G$
(the only incidences of $C_G$ already coloured are $(v_0,v_0v_1)$
and $(v_2,v_2v_1)$) as follows.
\begin{itemize}
\item We first colour the incidence $(v_1,v_1v_2)$, which has 
five already
coloured adjacent incidences.
\item We then cyclically colour the incidences of $C_G$ from $(v_2,v_2v_3)$
to $(v_{k-1},v_{k-1}v_{k-2})$.
This can be done since, doing so, every such incidence has four or five already
coloured adjacent incidences.
\item By Claim~\ref{claim:Halin}, the incidence $(v_{k-1},v_{k-1}v_0)$ has at most
five forbidden colours %already coloured adjacent incidences 
and can thus be coloured.
Similarly, thanks to Claim~\ref{claim:Halin}, we can also
colour the incidences $(v_0,v_0v_{k-1})$ and $(v_1,v_1v_0)$ (in that order).
\end{itemize}
\end{enumerate}
This completes the proof.
\end{proof}

%%%%%%%%%%%%%%%%%%%%%%%%%%%%%%%%%%%%%%%%%%%%%%%%%%%%%%%%%%%%%%%
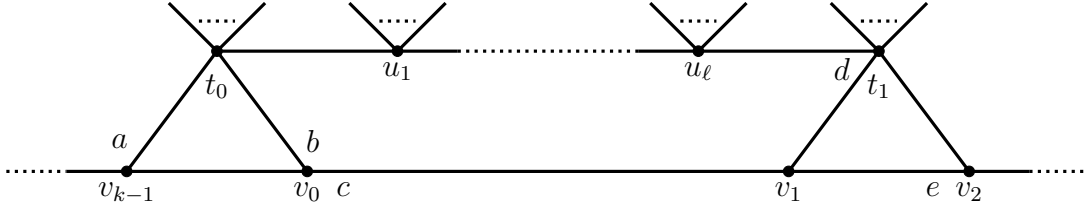
\begin{figure}
\begin{center}
\begin{tikzpicture}[domain=0:10,x=0.8cm,y=0.8cm]
%sommets
\SOMMET{2,0} \SOMMET{5,0} \SOMMET{13,0} \SOMMET{16,0}
\SOMMET{3.5,2} \SOMMET{6.5,2} \SOMMET{11.5,2} \SOMMET{14.5,2} 
%arêtes
\draw[very thick,dotted] (7.5,2) -- (10.5,2);
\draw[very thick] (3.5,2) -- (7.5,2);
\draw[very thick] (10.5,2) -- (14.5,2);
\draw[very thick] (2,0) -- (3.5,2);
\draw[very thick] (5,0) -- (3.5,2);
\draw[very thick] (13,0) -- (14.5,2);
\draw[very thick] (16,0) -- (14.5,2);
\draw[very thick] (1,0) -- (17,0);
\draw[very thick,dotted] (0,0) -- (1,0);
\draw[very thick,dotted] (17,0) -- (18,0);
%labels
\node[below] at (2,0) {$v_{k-1}$};
\node[below] at (5,0) {$v_{0}$};
\node[below] at (13,0) {$v_{1}$};
\node[below] at (16,0) {$v_{2}$};
\node[below] at (3.5,1.8) {$t_{0}$};
\node[below] at (14.5,1.8) {$t_{1}$};
\node[below] at (6.5,2) {$u_1$};
\node[below] at (11.5,2) {$u_\ell$};
% colours
%\node[right] at (3.8,1.7) {$d$};
\node[below] at (5.6,0) {$c$};
\node[right] at (4.8,0.5) {$b$};
\node[left] at (14.2,1.7) {$d$};
\node[below] at (15.4,0) {$e$};
\node[left] at (2.2,0.5) {$a$};
% adjacent incidences
\draw[very thick] (3.5,2) -- (2.7,2.8);
\draw[very thick] (3.5,2) -- (4.3,2.8);
\draw[very thick,dotted] (3.2,2.5) -- (3.8,2.5);
\draw[very thick] (6.5,2) -- (5.7,2.8);
\draw[very thick] (6.5,2) -- (7.3,2.8);
\draw[very thick,dotted] (6.2,2.5) -- (6.8,2.5);
\draw[very thick] (11.5,2) -- (10.7,2.8);
\draw[very thick] (11.5,2) -- (12.3,2.8);
\draw[very thick,dotted] (11.2,2.5) -- (11.8,2.5);
\draw[very thick] (14.5,2) -- (13.7,2.8);
\draw[very thick] (14.5,2) -- (15.3,2.8);
\draw[very thick,dotted] (14.2,2.5) -- (14.8,2.5);
\end{tikzpicture}
\caption{Colouring the subtree of $T_G$ in the proof of Lemma~\ref{lem:Halin2}.}
\label{fig:path-t1-t0}
\end{center}
\end{figure}
%%%%%%%%%%%%%%%%%%%%%%%%%%%%%%%%%%%%%%%%%%%%%%%%%%%%%%%%%%%%%%%%%%%%%%%%

%We still have to consider the case when $T_G$ is a star with,
%as observed at the beginning of the proof of Lemma~\ref{lem:Halin2},
%$\Delta(T_G)\le 4$.
%We thus have only two cases to consider, namely $T_G=K_{1,3}$ and $T_G=K_{1,4}$,
%that is the Halin graphs $K_4$ and $W_4$, respectively.

%We still have to consider the two particular Halin graphs $K_4$ and $W_4$.
%These two cases are dealt with in the following two lemmas.

The next lemma shows that the incidence choice number of $K_4$ is at most~6.

\begin{lemma}
$\ch_i(K_4)\leq 6$.
\label{lem:HalinK4}
\end{lemma}

\begin{proof}
Let $V(K_4)=\{v_0,v_1,v_2,v_3)$  and $L$ be any list assignment of $K_4$ such that
$|L(v_i,v_iv_j)|= 6$ for every incidence $(v_i,v_iv_j)$ of $K_4$.

The following claim will be useful for constructing
an $L$-list incidence colouring of $K_4$.

\begin{claim}\label{claim:HalinK4}
There exist $a\in L(v_1,v_1v_0)$, $b\in L(v_2,v_2v_0)$,
and $c\in L(v_3,v_3v_0)$
such that 
$$|L(v_0,v_0v_1)\cap\{a,b,c\}|\le 1.$$
\end{claim}

\begin{proof}
Let $A=L(v_1,v_1v_0)$, $B=L(v_2,v_2v_0)$ and $C=L(v_3,v_3v_0)$.
If $A\cap B\cap C\neq\emptyset$, then we set $a=b=c=\gamma$
for some $\gamma\in A\cap B\cap C$, so that $|L(v_0,v_0v_1)\cap\{a,b,c\}|\le 1$.

Otherwise, we consider two cases.
\begin{enumerate}
\item 
If $A$, $B$ and $C$ are pairwise disjoint, then at least two of them are distinct from
$L(v_0,v_0v_1)$, so that we can choose $a$, $b$ and $c$ in such a way that 
$|L(v_0,v_0v_1)\cap\{a,b,c\}|\le 1$.

\item 
Suppose now that $A\cap B\neq\emptyset$ 
(the cases $A\cap C\neq\emptyset$ and $B\cap C\neq\emptyset$ are similar).
We first set $a=b=\gamma$ for some $\gamma\in A\cap B$.
If $\gamma\in L(v_0,v_0v_1)$, then there exists $\varepsilon\in C\setminus L(v_0,v_0v_1)$
(since $(A\cap B)\cap C=\emptyset$) and we set $c=\varepsilon$, 
so that $|L(v_0,v_0v_1)\cap\{a,b,c\}|\le 1$.
If $\gamma\notin L(v_0,v_0v_1)$, then we set $e=\varepsilon$ for any $\varepsilon\in C$
and we also get $|L(v_0,v_0v_1)\cap\{a,b,c\}|\le 1$.
\end{enumerate}
This concludes the proof of Claim~\ref{claim:HalinK4}.
\end{proof}

We now construct an $L$-list incidence colouring $\sigma$ of $K_4$, by setting
first $\sigma(v_1,v_1v_0)=a$, $\sigma(v_2,v_2v_0)=b$ and $\sigma(v_3,v_3v_0)=c$,
where $a$, $b$ and $c$ are the values determined in the proof
of Claim~\ref{claim:HalinK4}.

We then consider two cases.
\begin{enumerate}
\item 
Suppose first that $|\{a,b,c\}|\le 2$ and assume $a=b$ (the cases $a=c$ and $b=c$ are similar).
We then colour the remaining uncoloured incidences as follows (see Figure~\ref{fig:HalinK4}(a)).
We first colour the incidences
$(v_3,v_3v_1)$, $(v_3,v_3v_2)$, $(v_2,v_2v_3)$, $(v_1,v_1v_3)$ and $(v_2,v_2v_1)$,
in that order.
This can be done since, doing so, every such incidence has at most five already coloured
adjacent incidences.
We then colour the incidences
$(v_1,v_1v_2)$, $(v_0,v_0v_3)$ and $(v_0,v_0v_2)$, in that order.
This can be done since, doing so, every such incidence has at most five
forbidden colours (recall that $a=b$).
We finally colour the incidence $(v_0,v_0v_1)$,
which has at least one available colour in its own list since,
by Claim~\ref{claim:HalinK4}, $|L(v_0,v_0v_1)\cap\{a,b,c\}|\le 1$).

%%%%%%%%%%%%%

%%%%%%%%%%%%%%%%%%%%%%%%%%%%%%%%%%%%%%%%%%%%%%%%%%%%%%%%%%%%%%%
\begin{figure}
\begin{center}
\begin{tikzpicture}[domain=0:10,x=0.8cm,y=0.8cm]
%sommets
\SOMMET{0,0} \SOMMET{3,0} \SOMMET{6,0} \SOMMET{3,3}
\node[above] at (3,3) {$v_0$};
\node[left] at (0,0) {$v_1$};
\node[below] at (3,0) {$v_2$};
\node[right] at (6,0) {$v_3$};
%arêtes
\draw[very thick] (0,0) -- (6,0);
\draw[very thick] (3,3) -- (0,0);
\draw[very thick] (3,3) -- (3,0);
\draw[very thick] (3,3) -- (6,0);
\draw[very thick] (0,0) to[bend right=45] (6,0);
% a b c
\node[left] at (0.4,0.5) {$a$};
\node[right] at (3,0.5) {$a$};
\node[right] at (5.6,0.5) {$c$};
% ordre incidences
\node[below] at (5.6,-0.4) {{\scriptsize $1$}};
\node[above] at (5.2,0) {{\scriptsize $2$}};
\node[above] at (3.8,0) {{\scriptsize $3$}};
\node[below] at (0.4,-0.4) {{\scriptsize $4$}};
\node[above] at (2.2,0) {{\scriptsize $5$}};
\node[above] at (0.8,0) {{\scriptsize $6$}};
\node[right] at (3.6,2.5) {{\scriptsize $7$}};
\node[right] at (3,2.2) {{\scriptsize $8$}};
\node[left] at (2.4,2.5) {{\scriptsize $9$}};
% legende
\node[below] at (3,-2) {(a) Case $1$ ($a=b$)};
\end{tikzpicture}
\hskip 0.5cm
\begin{tikzpicture}[domain=0:10,x=0.8cm,y=0.8cm]
%sommets
\SOMMET{0,0} \SOMMET{3,0} \SOMMET{6,0} \SOMMET{3,3}
\node[above] at (3,3) {$v_0$};
\node[left] at (0,0) {$v_1$};
\node[below] at (3,0) {$v_2$};
\node[right] at (6,0) {$v_3$};
%arêtes
\draw[very thick] (0,0) -- (6,0);
\draw[very thick] (3,3) -- (0,0);
\draw[very thick] (3,3) -- (3,0);
\draw[very thick] (3,3) -- (6,0);
\draw[very thick] (0,0) to[bend right=45] (6,0);
% a b c
\node[left] at (0.4,0.5) {$a$};
\node[right] at (3,0.5) {$b$};
\node[right] at (5.6,0.5) {{\scriptsize $7$}};
% ordre incidences
\node[below] at (5.6,-0.4) {{\scriptsize $6$}};
\node[above] at (5.2,0) {{\scriptsize $5$}};
\node[above] at (3.8,0) {{\scriptsize $4$}};
\node[below] at (0.4,-0.4) {{\scriptsize $1$}};
\node[above] at (2.2,0) {{\scriptsize $3$}};
\node[above] at (0.8,0) {{\scriptsize $2$}};
\node[right] at (3.6,2.5) {{\scriptsize $8$}};
\node[right] at (3,2.2) {{\scriptsize $9$}};
\node[left] at (2.4,2.5) {{\scriptsize $10$}};
% legende
\node[below] at (3,-2) {(b) Case $2(a)$};
\end{tikzpicture}

\begin{tikzpicture}[domain=0:10,x=0.8cm,y=0.8cm]
%sommets
\SOMMET{0,0} \SOMMET{3,0} \SOMMET{6,0} \SOMMET{3,3}
\node[above] at (3,3) {$v_0$};
\node[left] at (0,0) {$v_1$};
\node[below] at (3,0) {$v_2$};
\node[right] at (6,0) {$v_3$};
%arêtes
\draw[very thick] (0,0) -- (6,0);
\draw[very thick] (3,3) -- (0,0);
\draw[very thick] (3,3) -- (3,0);
\draw[very thick] (3,3) -- (6,0);
\draw[very thick] (0,0) to[bend right=45] (6,0);
% a b c
\node[left] at (0.4,0.5) {{\scriptsize $7$}};
\node[right] at (3,0.5) {{\scriptsize $4$}};
\node[right] at (5.6,0.5) {$c$};
% ordre incidences
\node[below] at (5.6,-0.4) {{\scriptsize $1$}};
\node[above] at (5.2,0) {{\scriptsize $2$}};
\node[above] at (3.8,0) {{\scriptsize $3$}};
\node[below] at (0.4,-0.4) {{\scriptsize $6$}};
\node[above] at (2.2,0) {{\scriptsize $5$}};
\node[above] at (0.8,0) {$d$};
\node[right] at (3.6,2.5) {$b$};
\node[right] at (3,2.2) {$a$};
\node[left] at (2.4,2.5) {{\scriptsize $8$}};
% legende
\node[below] at (3,-2) {(c) Case $2(b)$};
\end{tikzpicture}

\caption{Ordering the incidences of $K_4$ for the proof of Lemma~\ref{lem:HalinK4}.}
\label{fig:HalinK4}
\end{center}
\end{figure}
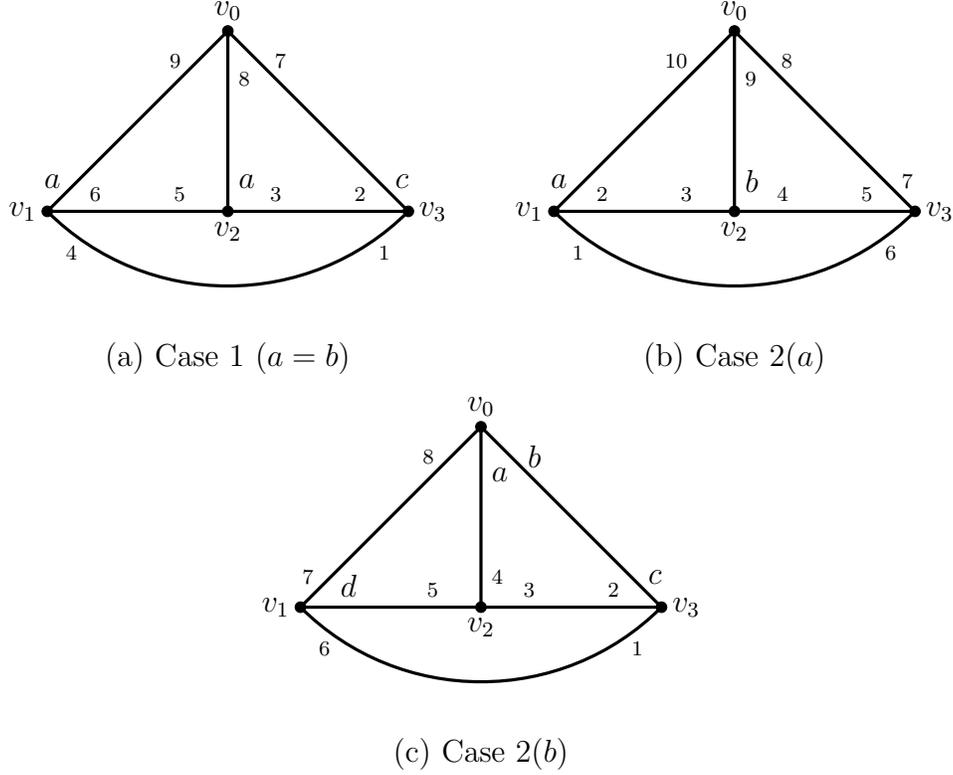
%%%%%%%%%%%%%%%%%%%%%%%%%%%%%%%%%%%%%%%%%%%%%%%%%%%%%%%%%%%%%%%%%%%%%%%%

\item 
Suppose now that $|\{a,b,c\}|=3$.
By symmetry and thanks to Claim~\ref{claim:HalinK4}, we may assume $L(v_0,v_0v_1)\cap\{a,b\}=\emptyset$,
without loss of generality.
We consider two subcases.
\begin{enumerate}
\item 
$|L(v_0,v_0v_2)\cap \{a,b\}|\le 1$ (or, similarly, $|L(v_0,v_0v_3)\cap \{a,b\}|\le 1$).\\
We first uncolour the incidence $(v_3,v_3v_0)$ (note that for any choice of $\sigma(v_3,v_3v_0)$, the statement of
Claim~\ref{claim:HalinK4} will be satisfied).
We then colour the remaining uncoloured incidences as follows (see Figure~\ref{fig:HalinK4}(b)).
We first colour the incidences
$(v_1,v_1v_3)$, 
$(v_1,v_1v_2)$, 
$(v_2,v_2v_1)$,
$(v_2,v_2v_3)$, 
$(v_3,v_3v_2)$, 
$(v_3,v_3v_1)$, 
$(v_3,v_3v_0)$ and $(v_0,v_0v_3)$, in that order.
This can be done since, doing so, every such incidence has at most five already coloured adjacent incidences.
We then colour the incidence $(v_0,v_0v_2)$, which has at most five forbidden colours since 
$|L(v_0,v_0v_2)\cap  \{a,b\}|\le 1$,
and the incidence $(v_0,v_0v_1)$,
which has also at most five forbidden colours since  
$|L(v_0,v_0v_1)\cap\{a,b,\sigma(v_3,v_3v_0)\}|\le 1$.

\item 
$\{a,b\}\subseteq (L(v_0,v_0v_2)\cap L(v_0,v_0v_3))$.\\
We first uncolour the incidences $(v_1,v_1v_0)$ and $(v_2,v_2v_0)$,
and set $\sigma(v_0,v_0v_2)=a$ and $\sigma(v_0,v_0v_3)=b$ (this is possible since $c\notin\{a,b\}$).

We claim that there exists a colour $d\in L(v_1,v_1v_2)$
such that $|L(v_1,v_1v_0)\cap\{b,d\}| \le 1.$
This is obviously the case if $b\notin L(v_1,v_1v_0)$.
Assume thus that $b\in L(v_1,v_1v_0)$.
If $b\in L(v_1,v_1v_2)$, then we can set $d=b$.
Otherwise, it suffices to choose any $d$ in $L(v_1,v_1v_2)\setminus L(v_1,v_1v_0)$.
%Since $a\in L(v_1,v_1v_0)$, we will have $d\neq a$ for any such $d$.
We then set $\sigma(v_1,v_1v_2)=d$.

We then colour the remaining uncoloured incidences as follows (see Figure~\ref{fig:HalinK4}(c)).
We first colour the incidences
%Finally, we colour the remaining uncoloured incidences
$(v_3,v_3v_1)$, $(v_3,v_3v_2)$, $(v_2,v_2v_3)$, 
$(v_2,v_2v_0)$, $(v_2,v_2v_1)$ and $(v_1,v_1v_3)$, in that order.
This can be done since, doing so, every such incidence has at most 
five already coloured adjacent incidences.
We then colour the incidences
$(v_1,v_1v_0)$, which has at most five forbidden colours since $|L(v_1,v_1v_0)\cap\{b,d\}| \le 1$),
 and $(v_0,v_0v_1)$, which has also at most five forbidden colours since
$L(v_0,v_0v_1)\cap\{a,b\}=\emptyset$.
\end{enumerate}

\end{enumerate}

This completes the proof.
\end{proof}

By Proposition~\ref{prop:bounds} and Lemmas~\ref{lem:Halin1}, \ref{lem:Halin2} and \ref{lem:HalinK4}, we 
 get:

\begin{theorem}
If $G$ is a  Halin graph, then
$$\left\{
  \begin{array}{ll}
     \ch_i(G)\leq 6, & \hbox{if $\Delta(G) \in\{3,4\}$ and $G\neq W_4$,} \\
     \ch_i(G)\leq 7, & \hbox{if $\Delta(G) = 5$ or $G=W_4$,} \\
     \ch_i(G)= \Delta(G)+1, & \hbox{otherwise.}
  \end{array}
\right.$$
\label{th:Halin}
\end{theorem}

%%%%%%%%%%%%%%%%%%%%%%%%%%%%%%%%%%%%%%%%%%%%%%%%%%%%%%%%%%%%%%%%%%%%%%%%%%%%%%%%%%%%%%%%%%%%%%%%%%%%%%%%%%%%%%%%%%%%%%%%%%%%%
%%%%%%%%%%%%%%%%%%%%%%%%%%%%%%%%%%%%%%%%%%%%%%%%%%%%%%%%%%%%%%%%%%%%%%%%%%%%%%%%%%%%%%%%%%%%%%%%%%%%%%%%%%%%%%%%%%%%%%%%%%%%%
%%%%%%%%%%%%%%%%%%%%%%%%%%%%%%%%% CACTUSES

\section{Cactuses}\label{sec:cactus}

A \emph{cactus} is a (planar) graph such that every vertex belongs to at most one cycle.
The \emph{corona} $G\odot K_1$ of a graph $G$ is the graph obtained from $G$
by adding one pendent neighbour to each vertex of~$G$.
A \emph{generalized corona} of a graph $G$ is a graph $G\odot pK_1$, for some
integer $p\ge 1$,  obtained from $G$
by adding $p$ pendent neighbours to each vertex of~$G$.
In particular, every generalized corona of a cycle is thus a cactus.

We give in this section an upper bound on the incidence choice number of cactuses.
In order to do that, we will first consider the case of generalized coronae of cycles.

For every integer $n\ge 3$, we let $V(C_n)=\{v_0,\dots,v_{n-1}\}$.
For every generalized corona $C_n\odot pK_1$ of the cycle $C_n$ and every vertex $v_i$ of $C_n$, 
$0\le i\le n-1$, we denote by $v_i^1,\dots,v_i^p$ the $p$~pendent neighbours of $v_i$.

Let $G=C_n\odot pK_1$, with $n\ge 3$ and $p\ge 1$, be a generalized corona of $C_n$,
and $L$ be any list assignment of $G$ such that $|L(v,vu)|= \Delta(G)+2$
for every incidence $(v,vu)$ of $G$. By colouring first the incidences of $C_n$,
then the uncoloured internal incidences of $v_0,\dots,v_{n-1}$, and finally
the external incidences of $v_0,\dots,v_{n-1}$, we can produce an 
$L$-list incidence colouring of $G$ since, doing so, every incidence has at most
$\Delta(G)+1$ already coloured adjacent incidences.
Therefore, $\ch_i(C_n\odot pK_1)\leq\Delta(C_n\odot pK_1)+2$ for every generalized corona $C_n\odot pK_1$.

The next lemma shows that we can decrease by~1 this bound whenever $p\geq 4$.
Note that by Proposition~\ref{prop:bounds}, in that case, the corresponding bound is tight.
Since it will be useful for studying the incidence choice number of cactuses,
the next lemma also considers the case when the two incidences of one pendent edge are pre-coloured,
and proves that an additional colour is needed in that case only when $n=3$ and $p\ge 3$.

\begin{lemma}
For every integers $n\ge 3$ and $p\ge 1$, 
$$\ch_i(C_n\odot pK_1)\le
\left\{
\begin{array}{ll}
\Delta(C_n\odot pK_1)+2=p+4, & \mbox{if $p\le 2$,} \\
\max(\Delta(C_n\odot pK_1)+1,7)=\max(p+3,7), & \mbox{otherwise.}
\end{array}
\right.$$
Moreover, for every list assignment $L$ of $C_n\odot pK_1$ with
$|L(v,vu)|=k$ for every incidence $(v,vu)$ of $C_n\odot pK_1$, 
$a\in L(v_0,v_0v_0^1)$ and $b\in L(v_0^1,v_0^1v_0)$, $a\neq b$,
there exists an $L$-incidence colouring $\sigma$ of $C_n\odot pK_1$ 
with $\sigma(v_0,v_0v_0^1)=a$ and $\sigma(v_0^1,v_0^1v_0)=b$
in each
of the following cases:
\begin{enumerate}
\item $p\le 2$ and $k\ge p+4$,
\item $n>3$, $p\ge 3$ and $k\ge \max(p+3,7)$,
\item $n=3$, $p\ge 3$ and $k\ge \max(p+3,8)$.
\end{enumerate}
\label{lem:coronae}
\end{lemma}

\begin{proof}
Since the proof when two incidences are pre-coloured is similar to the proof of
the general bound, we give these two proofs simultaneously, referring to the former
case as the pre-coloured case.
In the following, subscripts are always taken modulo $n$.

\medskip

We first consider the case $p\leq 2$.
Let $L$ be any list assignment of $C_n\odot pK_1$ such that
$|L(v,vu)|= p+4$ if $p\leq 2$ for every incidence $(v,vu)$ of $C_n\odot pK_1$,
and let $a\in L(v_0,v_0v_0^1)$ and $b\in L(v_0^1,v_0^1v_0)$, $a\neq b$.
We will construct an $L$-list incidence colouring $\sigma$ of $C_n\odot pK_1$ in three steps.
We first set $\sigma(v_0,v_0v_0^1)=a$ and $\sigma(v_0^1,v_0^1v_0)=b$, 
even if we are not in the pre-coloured case.

\begin{enumerate}
\item {\em Incidences of $C_n$}.\\
If $p=1$, there is only one edge incident to $v_0$  not belonging to $C_n$, and both its
incidences are already coloured.
If $p=2$, we claim that there exists $c\in L(v_{n-1},v_{n-1}v_0)$ 
such that $|L(v_0,v_0v_0^2)\cap \{a,b,c\}|\leq 2$ and we set $\sigma(v_{n-1},v_{n-1}v_0)=c$.
Indeed, 
if $\{a,b\}\not\subseteq L(v_0,v_0v_0^2)$, then $|L(v_0,v_0v_0^2)\cap \{a,b,c\}|\leq 2$ 
for any value of $c\in L(v_{n-1},v_{n-1}v_0)$.
Suppose now that  $\{a,b\}\subseteq L(v_0,v_0v_0^2)$. 
If $b\in L(v_{n-1},v_{n-1}v_0)$, then we set $c=b$. 
Otherwise, we set $c=\gamma$ for some $\gamma\in L(v_{n-1},v_{n-1}v_0)\setminus L(v_0,v_0v_0^2)$.

We then colour the remaining uncoloured incidences of $C_n$ cyclically,
from $(v_0,v_0v_{n-1})$ to $(v_{n-1},v_{n-1}v_{n-2})$,
which can be done since, doing so, every such incidence has at most $4<p+4$
already coloured adjacent incidences. 

\item {\em Uncoloured internal incidences of $v_i$, $0\le i\le n-1$}.\\
If $p=2$, we colour the incidence $(v_0,v_0v_0^2)$, which can be done
since it has at most~5 forbidden colours (recall that
$|L(v_0,v_0v_0^2)\cap \{\sigma(v_0,v_0v_0^1),\sigma(v_0^1,v_0^1v_0),\sigma(v_{n-1},v_{n-1}v_0)\}|\leq 2$
thanks to the previous step). 

Now, for each vertex $v_i$, $1\le i\le n-1$, we colour
the incidence $(v_i,v_iv_i^1)$, or the incidences 
$(v_i,v_iv_i^1)$ and $(v_i,v_iv_i^2)$, in that order, if $p=2$.
%, $\dots$, $(v_i,v_iv_i^p)$, in that order.
This can be done since, doing so, every such incidence $(v_i,v_iv_i^j)$, $1\le j\le p$,
has $j+3< p+4$ already coloured adjacent incidences.

\item {\em External incidences of $v_i$, $0\le i\le n-1$}.\\
We finally colour all uncoloured incidences of the form $(v_i^j,v_i^jv_i)$,
$0\le i\le n-1$, $1\le j\le p$, which can be done since
every such incidence has at most $p+2$ already coloured adjacent incidences.

\end{enumerate}

The above-constructed mapping $\sigma$ is clearly an $L$-list incidence colouring $\sigma$ of $C_n\odot pK_1$
with $\sigma(v_0,v_0v_0^1)=a$ and $\sigma(v_0^1,v_0^1v_0)=b$, as required.

\medskip

We now consider the case $p\geq 3$.
Let $L$ be any list assignment of $C_n\odot pK_1$ such that,
for every incidence $(v,vu)$ of $C_n\odot pK_1$,
$|L(v,vu)|= \max(p+3,7)$ if we are not in the pre-coloured case or $n>3$,
and $|L(v,vu)|= \max(p+3,8)$ otherwise.
Moreover, if we are in the pre-coloured case, then let
$a\in L(v_0,v_0v_0^1)$ and $b\in L(v_0^1,v_0^1v_0)$, $a\neq b$.

We will construct an $L$-list incidence colouring $\sigma$ of $C_n\odot pK_1$ in two steps.
If we are in the pre-coloured case, we first set $\sigma(v_0,v_0v_0^1)=a$ and $\sigma(v_0^1,v_0^1v_0)=b$.

\begin{enumerate}
\item {\em Incidences of $C_n$}.\\
We first construct a partial $L$-list incidence colouring $\sigma_0$ of $C_n\odot pK_1$,
fixing the colour of all incidences of $C_n$, and
satisfying the following property:
\begin{itemize}[leftmargin=1cm]
\item[(P)] For every $i$, $0\le i\le n-1$ (or $1\le i\le n-1$ if we are in the pre-coloured case), 
$$|L(v_i,v_iv_i^p) \cap \{\sigma_0(v_{i-1},v_{i-1}v_{i}),\sigma_0(v_{i+1},v_{i+1}v_{i})\}| \le 1.$$
Moreover, if we are in the pre-coloured case, then
$$|L(v_0,v_0v_0^p) \cap \{a,b,\sigma_0(v_{n-1},v_{n-1}v_{0}),\sigma_0(v_{1},v_{1}v_{0})\}| \le 2.$$
\end{itemize}

We proceed in two steps.
\begin{enumerate}

\item %1 
If we are in the pre-coloured case, then
we first claim that there exist $c\in L(v_1,v_1v_0)$ and $d\in L(v_{n-1},v_{n-1}v_0)$,
$c\neq a$, $d\neq a$, 
such that $|L(v_0,v_0v_0^p)\cap \{a,b,c,d\}|\leq 2$, and
 set $\sigma_0(v_1,v_1v_0)=c$ and $\sigma_0(v_{n-1},v_{n-1}v_0)=d$ (see Figure~\ref{fig:coronae-f1}).
 To see that, we consider two cases.

%%%%%%%%%%%%%%%%%%%%%%%%%%%%%%%%%%%%%%%%%%%%%%%%%%%%%%%%%%%%%%
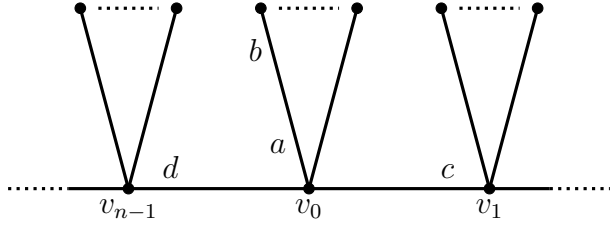
\begin{figure}
\begin{center}
\begin{tikzpicture}[domain=0:10,x=0.8cm,y=0.8cm]
%sommets
\SOMMET{-1,0} \SOMMET{2,0} \SOMMET{5,0} %\SOMMET{8,0} 
\SOMMET{-1.8,3} \SOMMET{-0.2,3} 
\SOMMET{1.2,3} \SOMMET{2.8,3} 
\SOMMET{4.2,3} \SOMMET{5.8,3} 
%arêtes
\draw[very thick,dotted] (-3,0) -- (-2,0);
\draw[very thick] (-2,0) -- (6,0);
\draw[very thick,dotted] (6,0) -- (7,0);
\draw[very thick,dotted] (-1.5,3) -- (-0.5,3);
\draw[very thick,dotted] (1.5,3) -- (2.5,3);
\draw[very thick,dotted] (4.5,3) -- (5.5,3);
\draw[very thick] (-1,0) -- (-1.8,3);
\draw[very thick] (-1,0) -- (-0.2,3);
\draw[very thick] (2,0) -- (1.2,3);
\draw[very thick] (2,0) -- (2.8,3);
\draw[very thick] (5,0) -- (4.2,3);
\draw[very thick] (5,0) -- (5.8,3);
%labels
\node[below] at (-1,0) {$v_{n-1}$};
\node[below] at (2,0) {$v_{0}$};
\node[below] at (5,0) {$v_{1}$};
%\node[below] at (8,0) {$v_{2}$};
%couleurs
\node[above] at (4.3,0) {$c$};
\node[above] at (-0.3,0) {$d$};
\node[left] at (1.8,0.7) {$a$};
\node[left] at (1.4,2.3) {$b$};

\end{tikzpicture}
\caption{Configuration for the proof of Lemma~\ref{lem:coronae}, pre-coloured case.}
\label{fig:coronae-f1}
\end{center}
\end{figure}
%%%%%%%%%%%%%%%%%%%%%%%%%%%%%%%%%%%%%%%%%%%%%%%%%%%%%%%%%%%%%%%%%%%%%%%%
 
\begin{enumerate}
\item $|\{a,b\}\cap L(v_0,v_0v_0^p)[ \leq 1$.\\
In that case, it suffices to choose $c$ and $d$ in such a way that  
$|\{c,d\}\cap L(v_0,v_0v_0^p)| \leq 1$.
This can be done since either $(L(v_{n-1},v_{n-1}v_0)\cap L(v_1,v_1v_0))\setminus\{a\}\neq \emptyset$,
in which case we choose $c=d=\gamma$ for some 
$\gamma\in (L(v_{n-1},v_{n-1}v_0)\cap L(v_1,v_1v_0))\setminus\{a\}$,
or $(L(v_{n-1},v_{n-1}v_0)\cap L(v_1,v_1v_0))\setminus\{a\}=\emptyset$,
which implies 
$$|L(v_{n-1},v_{n-1}v_0)\cup L(v_1,v_1v_0)|\ge \max(2(p+2),12),$$ 
and we can choose $c$ and $d$ in such a way that
$|\{c,d\}\cap L(v_0,v_0v_0^p)[ \leq 1$.

\item $\{a,b\}\subseteq L(v_0,v_0v_0^p)$.\\
If $b\in L(v_1,v_1v_0)$, then we set $c=b$. 
Otherwise, we set $c=\gamma$ for some $\gamma\in L(v_1,v_1v_0)\setminus L(v_0,v_0v_0^p)$.
Similarly, if $b\in L(v_{n-1},v_{n-1}v_0)$ then we set $d=b$. Otherwise, we set $d=\delta$ for 
some $\delta\in L(v_{n-1},v_{n-1}v_0)\setminus L(v_0,v_0v_0^p)$.

\end{enumerate}

In all cases, we get $|L(v_0,v_0v_0^p)\cap \{a,b,c,d\}|\leq 2$.

\medskip

In both cases (pre-coloured or not), we are going to colour some incidences of $C_n$, 
in such a way that for every $i$, $0\le i\le n-1$ 
(or $1\le i\le n-1$ if we are in the pre-coloured case), we have the following property:

\begin{itemize}[leftmargin=1cm]
\item[(P')] Either
$\sigma_0(v_{i-1},v_{i-1}v_{i}) = \sigma_0(v_{i+1},v_{i+1}v_{i})$, or
 one of $\sigma_0(v_{i-1},v_{i-1}v_{i})$, $\sigma_0(v_{i+1},v_{i+1}v_{i})$ only is set
and, in that case, the assigned colour does not belong to $L(v_i,v_iv_i^p)$.
\end{itemize}

For every such $i$, we denote by $\alpha_i$ the colour assigned to one or both external
incidences of $v_i$.
If we are in the pre-coloured case, we first deal with the external incidences
of $v_1$ and~$v_{n-1}$.

%%%%%%%%%%%%%%%%%%%%%%%%%%%%%%%%%%%%%%%%%%%%%%%%%%%%%%%%%%%%%%
\begin{figure}
\begin{center}
\begin{tikzpicture}[domain=0:10,x=0.8cm,y=0.8cm]
%sommets
\SOMMET{-1,0} \SOMMET{2,0} \SOMMET{5,0} \SOMMET{8,0} 
\SOMMET{-1.8,3} \SOMMET{-0.2,3} 
\SOMMET{1.2,3} \SOMMET{2.8,3} 
\SOMMET{4.2,3} \SOMMET{5.8,3} 
\SOMMET{7.2,3} \SOMMET{8.8,3} 
%arêtes
\draw[very thick,dotted] (-3,0) -- (-2,0);
\draw[very thick] (-2,0) -- (9,0);
\draw[very thick,dotted] (9,0) -- (10,0);
\draw[very thick,dotted] (-1.5,3) -- (-0.5,3);
\draw[very thick,dotted] (1.5,3) -- (2.5,3);
\draw[very thick,dotted] (4.5,3) -- (5.5,3);
\draw[very thick,dotted] (7.5,3) -- (8.5,3);
\draw[very thick] (-1,0) -- (-1.8,3);
\draw[very thick] (-1,0) -- (-0.2,3);
\draw[very thick] (2,0) -- (1.2,3);
\draw[very thick] (2,0) -- (2.8,3);
\draw[very thick] (5,0) -- (4.2,3);
\draw[very thick] (5,0) -- (5.8,3);
\draw[very thick] (8,0) -- (7.2,3);
\draw[very thick] (8,0) -- (8.8,3);
%labels
\node[below] at (-1,0) {$v_{n-1}$};
\node[below] at (2,0) {$v_{0}$};
\node[below] at (5,0) {$v_{1}$};
\node[below] at (8,0) {$v_{2}$};
%couleurs
\node[above] at (4.3,0) {$c$};
\node[above] at (-0.3,0) {$d$};
\node[left] at (1.8,0.7) {$a$};
\node[left] at (1.4,2.3) {$b$};

\end{tikzpicture}
\caption{Colouring the external incidences of $v_1$ ($v_{n-1}=v_2$ if $n=3$), pre-coloured case.}
\label{fig:coronae-f2}
\end{center}
\end{figure}
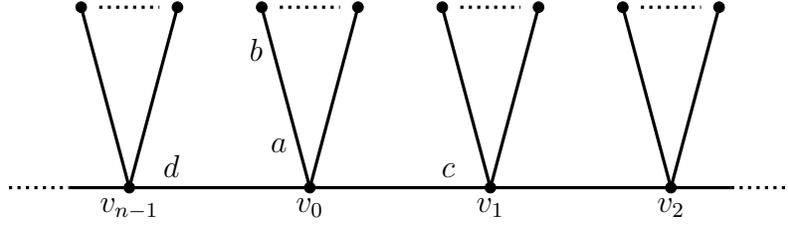
%%%%%%%%%%%%%%%%%%%%%%%%%%%%%%%%%%%%%%%%%%%%%%%%%%%%%%%%%%%%%%%%%%%%%%%%

\begin{itemize}
\item {\em External incidences of $v_1$, pre-coloured case} (see Figure~\ref{fig:coronae-f2}).\\
Let $L'(v_{0},v_{0}v_{1})=L(v_{0},v_{0}v_{1})\setminus\{a,b,c,d\}$,
and
$$L'(v_{2},v_{2}v_{1})
=\left\{
    \begin{array}{ll}
       L(v_{2},v_{2}v_{1})\setminus\{c,d\}, & \hbox{if $n=3$,} \\
       L(v_{2},v_{2}v_{1})\setminus\{c\}, & \hbox{otherwise.}
    \end{array}
 \right.
$$
If $L'(v_{0},v_{0}v_{1}) \cap L'(v_{2},v_{2}v_{1}) \neq \emptyset$,
then we set $\sigma_0(v_{0},v_{0}v_{1}) = \sigma_0(v_{2},v_{2}v_{1}) = \alpha_1$
for some $\alpha_1 \in L'(v_{0},v_{0}v_{1}) \cap L'(v_{2},v_{2}v_{1})$.
Otherwise, we consider two cases.
\begin{itemize}
\item If $n=3$, then  $|L(v,vu)|\ge \max(p+3,8)$ %%%% 8
for every incidence $(v,vu)$ of $C_n\odot pK_1$,
which implies $|L'(v_{0},v_{0}v_{1})| \ge \max(p-1,4)$ %%%% 4
and $|L'(v_{2},v_{2}v_{1})| \ge \max(p+1,6)$, %%%% 6
so that $|L'(v_{0},v_{0}v_{1}) \cup L'(v_{2},v_{2}v_{1})| \ge \max(2p,10)$. %%%% 10
Therefore, either there exists some colour $\alpha_1\in L'(v_{0},v_{0}v_{1})\setminus L(v_1,v_1v_1^p)$,
in which case we set $\sigma_0(v_{0},v_{0}v_{1})=\alpha_1$,
or there exists some colour $\alpha_1\in L'(v_{2},v_{2}v_{1})\setminus L(v_1,v_1v_1^p)$,
and we set $\sigma_0(v_{2},v_{2}v_{1})=\alpha_1$.

\item If $n\geq 4$, then $|L(v,vu)|\ge \max(p+3,7)$ %%%% 7
for every incidence $(v,vu)$ of $C_n\odot pK_1$,
which implies $|L'(v_{0},v_{0}v_{1})| \ge \max(p-1,3)$ %%%% 3
and $|L'(v_{2},v_{2}v_{1})| \ge \max(p+2,6)$, %%%% 6
so that $|L'(v_{0},v_{0}v_{1}) \cup L'(v_{2},v_{2}v_{1})| \ge \max(2p+1,9)$. %%%% 9
Therefore, either there exists some colour $\alpha_1\in L'(v_{0},v_{0}v_{1})\setminus L(v_1,v_1v_1^p)$,
in which case we set $\sigma_0(v_{0},v_{0}v_{1})=\alpha_1$,
or there exists some colour $\alpha_1\in L'(v_{2},v_{2}v_{1})\setminus L(v_1,v_1v_1^p)$,
and we set $\sigma_0(v_{2},v_{2}v_{1})=\alpha_1$.
\end{itemize}

\item {\em External incidences of $v_{n-1}$, pre-coloured case}.\\
Let $L'(v_{0},v_{0}v_{n-1})=L(v_{0},v_{0}v_{n-1})\setminus\{a,b,c,d,\alpha_1\}$,
and
$$L'(v_{n-2},v_{n-2}v_{n-1})
=\left\{
    \begin{array}{lll}
        L(v_{n-2},v_{n-2}v_{n-1})\setminus\{c,d,\alpha_1\}, & \hbox{if $n=3$,} \\
        L(v_{n-2},v_{n-2}v_{n-1})\setminus\{d,\alpha_1\}, & \hbox{if $n=4$,} \\
        L(v_{n-2},v_{n-2}v_{n-1})\setminus\{d\}, & \hbox{otherwise.}
     \end{array}
\right.
$$
If $d\not\in L(v_{n-1},v_{n-1}v_{n-1}^p)$, then we set 
$\sigma_0(v_{0},v_{0}v_{n-1})=\alpha_{n-1}$ 
for some $\alpha_{n-1}\in L'(v_{0},v_{0}v_{n-1})$ and we are done.

Suppose now that $d\in L(v_{n-1},v_{n-1}v_{n-1}^p)$.
If $L'(v_{n-2},v_{n-2}v_{n-1}) \cap L'(v_{0},v_{0}v_{n-1}) \neq \emptyset$,
then we set $\sigma_0(v_{n-2},v_{n-2}v_{n-1}) = \alpha_{n-1}$
and $\sigma_0(v_{0},v_{0}v_{n-1}) = \alpha_{n-1}$
for some $\alpha_{n-1} \in L'(v_{n-2},v_{n-2}v_{n-1}) \cap L'(v_{0},v_{0}v_{n-1})$.
Otherwise, we consider two cases.

\begin{itemize}
\item If $n=3$ (and thus, $(v_{n-2},v_{n-2}v_{n-1})=(v_{1},v_{1}v_{2})$), then $|L(v,vu)|\ge \max(p+3,8)$ %%%% 8
for every incidence $(v,vu)$ of $C_3\odot pK_1$,
which implies $|L'(v_{1},v_{1}v_{2})| \ge \max(p,5)$ %%%% 5
and $|L'(v_{0},v_{0}v_{2})| \ge \max(p-2,3)$, %%%% 3
so that $|L'(v_{1},v_{1}v_{2}) \cup L'(v_{0},v_{0}v_{2})| \ge \max(2p-2,8)$. %%%% 8
Note that $L'(v_{1},v_{1}v_{2}) \cup L'(v_{0},v_{0}v_{2}) \neq L(v_{2},v_{2}v_{2}^p)$
since $d\in L(v_{2},v_{2}v_{2}^p)$    
and $d \not\in L'(v_{1},v_{1}v_{2}) \cup L'(v_{0},v_{0}v_{2})$.
Therefore, 
either there exists some colour $\alpha_{2}\in L'(v_{1},v_{1}v_{2})\setminus L(v_{2},v_{2}v_{2}^p)$,
in which case we set $\sigma_0(v_{1},v_{1}v_{2})=\alpha_{2}$,
or there exists some colour $\alpha_{2}\in L'(v_{0},v_{0}v_{2})\setminus L(v_{2},v_{2}v_{2}^p)$,
and we set $\sigma_0(v_{0},v_{0}v_{2})=\alpha_{2}$.

%%%%%%%%%%%%%%%%%%%%%%%%%%%%%%%%%%%%%%%%%%%%%%%%%%%%%%%%%%%%%%
\begin{figure}
\begin{center}
\begin{tikzpicture}[domain=0:10,x=0.8cm,y=0.8cm]
%sommets
\SOMMET{-4,0} \SOMMET{-1,0} \SOMMET{2,0} \SOMMET{5,0} \SOMMET{8,0} 
\SOMMET{-4.8,3} \SOMMET{-3.2,3} 
\SOMMET{-1.8,3} \SOMMET{-0.2,3} 
\SOMMET{1.2,3} \SOMMET{2.8,3} 
\SOMMET{4.2,3} \SOMMET{5.8,3} 
\SOMMET{7.2,3} \SOMMET{8.8,3} 
%arêtes
\draw[very thick,dotted] (-6,0) -- (-5,0);
\draw[very thick] (-5,0) -- (9,0);
\draw[very thick,dotted] (9,0) -- (10,0);
\draw[very thick,dotted] (-4.5,3) -- (-3.5,3);
\draw[very thick,dotted] (-1.5,3) -- (-0.5,3);
\draw[very thick,dotted] (1.5,3) -- (2.5,3);
\draw[very thick,dotted] (4.5,3) -- (5.5,3);
\draw[very thick,dotted] (7.5,3) -- (8.5,3);
\draw[very thick] (-4,0) -- (-4.8,3);
\draw[very thick] (-4,0) -- (-3.2,3);
\draw[very thick] (-1,0) -- (-1.8,3);
\draw[very thick] (-1,0) -- (-0.2,3);
\draw[very thick] (2,0) -- (1.2,3);
\draw[very thick] (2,0) -- (2.8,3);
\draw[very thick] (5,0) -- (4.2,3);
\draw[very thick] (5,0) -- (5.8,3);
\draw[very thick] (8,0) -- (7.2,3);
\draw[very thick] (8,0) -- (8.8,3);
%labels
\node[below] at (-4,0) {$v_{n-2}$};
\node[below] at (-1,0) {$v_{n-1}$};
\node[below] at (2,0) {$v_{0}$};
\node[below] at (5,0) {$v_{1}$};
\node[below] at (8,0) {$v_{2}$};
%couleurs
\node[above] at (4.3,0) {$c$};
\node[above] at (-0.3,0) {$d$};
\node[above] at (2.7,0) {$\alpha_1?$};
\node[above] at (7.3,0) {$\alpha_1?$};
\node[left] at (1.8,0.7) {$a$};
\node[left] at (1.4,2.3) {$b$};

\end{tikzpicture}
\caption{Colouring the external incidences of $v_{n-1}$ ($v_{n-2}=v_2$ if $n=4$), pre-coloured case.
At least one of the incidences $(v_0,v_0v_1)$ or $(v_2,v_2v_1)$ is coloured with $\alpha_1$.}
\label{fig:coronae-f3}
\end{center}
\end{figure}
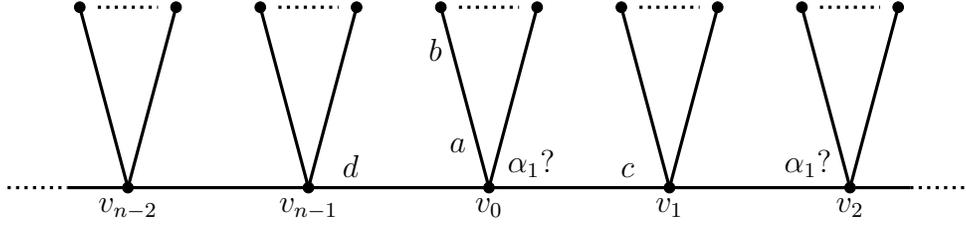
%%%%%%%%%%%%%%%%%%%%%%%%%%%%%%%%%%%%%%%%%%%%%%%%%%%%%%%%%%%%%%%%%%%%%%%%

\item If $n\geq 4$ (see Figure~\ref{fig:coronae-f3}), then $|L(v,vu)|\ge \max(p+3,7)$ %%%% 7
for every incidence $(v,vu)$ of $C_n\odot pK_1$,
which implies $|L'(v_{n-2},v_{n-2}v_{n-1})| \ge \max(p+1,5)$ %%%% 5
and $|L'(v_{0},v_{0}v_{n-1})| \ge \max(p-2,2)$, %%%% 2
so that $|L'(v_{n-2},v_{n-2}v_{n-1}) \cup L'(v_{0},v_{0}v_{n-1})| \ge \max(2p-1,7)$. %%%% 7
As in the previous case, $L'(v_{n-2},v_{n-2}v_{n-1})) \cup L'(v_{0},v_{0}v_{n-1}) \neq L(v_{n-1},v_{n-1}v_{n-1}^p)$
since $d\in L(v_{n-1},v_{n-1}v_{n-1}^p)$    
and $d \not\in L'(v_{n-2},v_{n-2}v_{n-1}) \cup L'(v_{0},v_{0}v_{n-1})$.
Therefore, either there exists some colour $\alpha_{n-1}\in L'(v_{n-2},v_{n-2}v_{n-1})\setminus L(v_{n-1},v_{n-1}v_{n-1}^p)$,
in which case we set $\sigma_0(v_{n-2},v_{n-2}v_{n-1})=\alpha_{n-1}$,
or there exists some colour $\alpha_{n-1}\in L'(v_{0},v_{0}v_{n-1})\setminus L(v_{n-1},v_{n-1}v_{n-1}^p)$,
and we set $\sigma_0(v_{0},v_{0}v_{n-1})=\alpha_{n-1}$.
\end{itemize}

\end{itemize}

For constructing the partial colouring $\sigma_0$,
we proceed sequentially, 
from $i=2$ to $i=n-2$ if we are in the pre-coloured case and $n\neq 3$
(note that $\sigma_0$ is already constructed if $n=3$),
or from $i=0$ to $i=n-1$ otherwise. 

For each such $i$, let
$$L'(v_{i-1},v_{i-1}v_{i})=L(v_{i-1},v_{i-1}v_{i})\setminus\{\alpha_{i-2},\alpha_{i-1},\alpha_{i+1}\},\ \mbox{and}$$ 
$$L'(v_{i+1},v_{i+1}v_{i})=L(v_{i+1},v_{i+1}v_{i})\setminus\{\alpha_{i-1},\alpha_{i+1},\alpha_{i+2}\},$$
if we are not in the pre-coloured case, or
$$L'(v_{i-1},v_{i-1}v_{i})
=\left\{
    \begin{array}{ll}
       L(v_{i-1},v_{i-1}v_{i})\setminus\{c, \alpha_{1},\alpha_{3}\}, & \hbox{if $i=2$,} \\
       L(v_{i-1},v_{i-1}v_{i})\setminus\{\alpha_{i-2},\alpha_{i-1},\alpha_{i+1}\}, & \hbox{otherwise,}
    \end{array}
 \right.
\mbox{and}$$
$$L'(v_{i+1},v_{i+1}v_{i})
=\left\{
    \begin{array}{ll}
       L(v_{i+1},v_{i+1}v_{i})\setminus\{\alpha_{n-3},\alpha_{n-1},d\}, & \hbox{if $i=n-2$,} \\
       L(v_{i+1},v_{i+1}v_{i})\setminus\{\alpha_{i-1},\alpha_{i+1},\alpha_{i+2}\}, & \hbox{otherwise,}
    \end{array}
 \right.
$$
if we are in the pre-coloured case.

Note here that when proceeding with $i$, the colour $\alpha_{i-2}$ (resp. $\alpha_{i-1}$, $\alpha_{i+1}$, $\alpha_{i+2}$)
is defined only if $i\ge 2$ (resp. $i\ge 1$, $i\le n-1$, $i\le n-2$).

%We then construct $\sigma_0$ as follows.
If $L'(v_{i-1},v_{i-1}v_{i}) \cap L'(v_{i+1},v_{i+1}v_{i}) \neq \emptyset$,
we set $\sigma_0(v_{i-1},v_{i-1}v_{i}) = \sigma_0(v_{i+1},v_{i+1}v_{i}) = \alpha_i$ 
for some $\alpha_i \in L'(v_{i-1},v_{i-1}v_{i}) \cap L'(v_{i+1},v_{i+1}v_{i})$.

Otherwise, 
since $|L(v,vu)|\ge \max(p+3,7)$ %%%% 7
for every incidence $(v,vu)$ of $C_n\odot pK_1$,
which implies $|L'(v_{i-1},v_{i-1}v_{i})| \ge \max(p,4)$ %%%% 4
and $|L'(v_{i+1},v_{i+1}v_{i})| \ge \max(p,4)$, %%%% 4
so that $|L'(v_{i-1},v_{i-1}v_{i}) \cup L'(v_{i+1},v_{i+1}v_{i})| \ge \max(2p,8)$, %%%% 8
either there exists some colour $\alpha_i\in L'(v_{i-1},v_{i-1}v_{i})\setminus L(v_i,v_iv_i^p)$,
in which case we set $\alpha_0(v_{i-1},v_{i-1}v_{i})=\alpha_i$,
or there exists some colour $\alpha_i\in L'(v_{i+1},v_{i+1}v_{i})\setminus L(v_i,v_iv_i^p)$,
and we set $\alpha_0(v_{i+1},v_{i+1}v_{i})=\alpha_i$.

By construction, the partial $L$-list incidence colouring $\sigma_0$ clearly satisfies
Property (P').

\item %2
We now colour the remaining uncoloured incidences of $C_n$, which can be done
since every such incidence has at most four already coloured adjacent incidences.
Thanks to Property (P'), and since at least one of the external incidences
of each vertex $v_i$ has been coloured in the previous step,
the partial $L$-list incidence colouring $\sigma_0$ 
thus obtained  satisfies Property (P).

\end{enumerate}

\item 
We now extend the partial $L$-list incidence colouring $\sigma_0$
to an $L$-list incidence colouring $\sigma$ of $C_n\odot pK_1$.
The only remaining uncoloured incidences are
the internal and external incidences of pendent vertices
(except $(v_0,v_0v_0^1)$ and $(v_0^1,v_0^1v_0)$ if we are
in the pre-coloured case, which are already coloured by $a$ and $b$, respectively).

We proceed as follows.
If we are in the pre-coloured case, then we first colour
the incidences $(v_0,v_0v_0^2)$, $\dots$, $(v_0,v_0v_0^p)$, in that order,
otherwise we first colour 
the incidences $(v_0,v_0v_0^1)$, $\dots$, $(v_0,v_0v_0^p)$, in that order.
Then, for each vertex $v_i$, $1\le i\le n-1$, we colour
the incidences $(v_i,v_iv_i^1)$, $\dots$, $(v_i,v_iv_i^p)$, in that order.
This can be done since, doing so,
\begin{enumerate}
\item every incidence $(v_i,v_iv_i^j)$, $1\le j\le p-1$,
has $j+3\le p+2$ already coloured adjacent incidences
(recall that $|L(v_i,v_iv_i^j)|\ge p+3$), and
\item thanks to Property (P)
(and to the fact that $|L(v_0,v_0v_0^p)\cap \{a,b,c,d\}|\leq 2$ if we are in the precoloured case), 
the incidence $(v_i,v_iv_i^p)$
has at most $p+2$ forbidden colours.
\end{enumerate}
We finally colour all the uncoloured incidences of the form $(v_i^j,v_i^jv_i)$,
$0\le i\le n-1$, $1\le j\le p$, which can be done since
every such incidence has  $p+2$ already coloured adjacent incidences.

\end{enumerate}

This completes the proof.
\end{proof}

We are now able to prove the main result of this section.
Let $G$ be a cactus, and $C$ be a cycle in $G$.
We say that $C$ is a {\em maximal cycle} if $C$ contains
a vertex $v$ with $\deg_G(v)=\Delta(G)$.

\begin{theorem}
Let $G$ be a cactus which is neither a tree nor a cycle. We then have
$$\ch_i(G)\leq
\left\{
    \begin{array}{ll}
      \Delta(G)+2, & \hbox{if $\Delta(G)=3$,}\\
      \Delta(G)+1, & \hbox{if $\Delta(G)=4$ and $G$ has no maximal cycle,}\\
      \Delta(G)+2, & \hbox{if $\Delta(G)=4$ and $G$ has a maximal cycle,}\\
      \max(\Delta(G)+1, 7), & \hbox{if $\Delta(G)\geq 5$ and $G$ has at most one maximal $3$-cycle,} \\
      \max(\Delta(G)+1, 8), & \hbox{otherwise.}
    \end{array}
\right.$$
\label{th:cactus}
\end{theorem}

\begin{proof}
Let $L$ be a list assignment of $G$ such that
$|L(v,vu)|=k$ for every incidence $(v,vu)$ of $G$, where $k$ is the value claimed
in the statement of the theorem.

Let $C_1,\dots,C_\ell$, $\ell\geq 1$, denote the cycles in $G$,
and $M$ denote the graph obtained from $G$ by contracting each cycle $C_i$
into a vertex $c_i$. 
The graph $M$ is clearly a tree.
Let us call each vertex $c_i$ in $M$ a {\em cycle vertex} and each other vertex in $M$, if any, a {\em normal vertex}.
Moreover, if $G$ contains a maximal 3-cycle, we assume without loss of generality that this cycle is $C_1$.
We now order all the vertices of $M$, starting with $c_1$, in such a way that each vertex $v\neq c_1$ has
exactly one neighbour among the vertices preceding $v$ in the order
(this can be done since $M$ is a tree).

We now colour the incidences of $G$ according to the ordering of the vertices of $M$
as follows. Let $v$ be the vertex of $M$ to be treated.
We have two cases to consider.

\begin{enumerate}
\item {\em $v$ is a cycle vertex of $M$.}\\
Let $v=c_i$, $1\le i\le\ell$.
We then colour all the incidences of the subgraph $H_i$ of $G$
induced by the vertices of the cycle $C_i$ and their neighbours.
The subgraph $H_i$ is a subgraph of some generalized corona and thus,
thanks to Observation~\ref{obs:subgraph} and Lemma~\ref{lem:coronae},
all the incidences of $H_i$ can be coloured.

\item {\em $v$ is a normal vertex of $M$.}\\ 
In that case, $v$ is also a vertex in $G$.
We colour the uncoloured internal incidences of $v$, if any,
and then the uncoloured external incidences of $v$, if any, in that order.
This can be done since, doing so, every such incidence has
at most $\Delta(G)$ already coloured adjacent incidences.
\end{enumerate}

This concludes the proof.
\end{proof}

Note that thanks to Proposition~\ref{prop:bounds}, the bound given
in Theorem~\ref{th:cactus} is tight for every cactus $G$ such that 
$\Delta(G)\geq 7$,
or $\Delta(G)=6$ and $G$ has at most one maximal $3$-cycle,
or $\Delta(G)=4$ and $G$ has no maximal cycle.

%%%%%%%%%%%%%%%%%%%%%%%%%%%%%%%%%%%%%%%%%%%%%%%%%%%%%%%%%%%%%%%%%%%%%%%%%%%%%%%%%%%%%%%%%%%%%%%%%%%%%%%%%%%%%%%%%%%%%%%%%%%%%
%%%%%%%%%%%%%%%%%%%%%%%%%%%%%%%%%%%%%%%%%%%%%%%%%%%%%%%%%%%%%%%%%%%%%%%%%%%%%%%%%%%%%%%%%%%%%%%%%%%%%%%%%%%%%%%%%%%%%%%%%%%%%
%%%%%%%%%%%%%%%%%%%%%%%%%%%%%%%%% CUBIC GRAPHS

%\section{Hamiltonian cubic and strictly subcubic graphs}\label{sec:cubic}
\section{Hamiltonian cubic graphs}\label{sec:cubic}

By Proposition~\ref{prop:bounds-2}, we know that $\ch_i(G)\leq 7$ for
every graph with maximum degree~3. 
We prove in this section that this bound can be decreased to~6 for
Hamiltonian cubic graphs. 
(Recall that by the result of Maydanskyi~\cite{M05},
$\chi_i(G)\leq 5$ for every cubic graph.)

Let $G$ be a Hamiltonian cubic graph of order $n$ ($n$ is necessarily even)
and $C_G=v_0v_1\dots v_{n-1}v_0$ be a Hamilton cycle in $G$.
The set of edges $F=E(G)\setminus E(C_G)$ is thus a perfect matching.
We denote by $F_G$ the subgraph of $G$ induced by $F$.
Let $v_i$, $0\le i\le n-1$, be a vertex of $G$.
The {\em matched vertex} of $v_i$ (with respect to $C_G$) is the unique vertex $v_j$ such that $v_iv_j\in F$.
The {\em antipodal vertex} of $v_i$ (with respect to $C_G$) is the vertex $v_{i+\frac{n}{2}}$ (subscripts are taken modulo $n$).
Two vertices $x$ and $y$ of $G$ are {\em consecutive} (with respect to $C_G$) if there exists some $i$, $0\le i\le n-1$,
such that $\{x,y\}=\{v_i,v_{i+1}\}$ (subscripts are taken modulo $n$).

%The following lemma has been proved by Shiu, Lam and Chen~\cite{SLC02}.
%
%\begin{lemma}[Shiu, Lam and Chen~\cite{SLC02}]
%If $G$ is a cubic Hamiltonian graph and $C_G$ a Hamilton cycle of $G$,
%then either $G$ contains two consecutive vertices whose matched
%vertices are not consecutive or each vertex is matched with its antipodal vertex (with respect to $C_G$).
%\label{lem:cub}
%\end{lemma}

We first prove the following easy lemma.

\begin{lemma}
If $G$ is a Hamiltonian cubic graph of order $n\ge 6$
and $C_G=v_0v_1\dots v_{n-1}v_0$ a Hamilton cycle in $G$,
then there exists a vertex $v_i$ in $G$, $0\le i\le n-1$,
such that $v_{i+2}$ is not the matched vertex of $v_i$.
\label{lem:cub2}
\end{lemma}

\begin{proof}
If $v_2$ is not the matched vertex of $v_0$ then $v_0$ satisfies the required property.
Otherwise, since $n\ge 6$, $v_2$ satisfies the required property. 
\end{proof}

We now prove the main result of this section.

\begin{theorem}
For every Hamiltonian cubic graph $G$, $\ch_i(G)\leq 6$.
\label{th:HamiltonianCubic}
\end{theorem}

\begin{proof}
Let $G$ be a Hamiltonian cubic graph, $C_G=v_0v_1\dots v_{n-1}v_0$ be a Hamilton cycle in $G$, 
and $L$ be any list assignment of~$G$ such that
$|L(v,vu)|=6$ for every incidence $(v,vu)$ of~$G$.
In the following, subscripts are always taken modulo $n$.

Note first that if $n=4$, then $G=K_4$ and the result follows from Lemma~\ref{lem:HalinK4}.
We thus assume $n\geq 6$.
Each vertex~$v_i$, $0\le i\le n-1$, has three neighbours in $G$, namely $v_{i-1}$, $v_{i+1}$ and 
the matched vertex $v'_i=v_{j}$ of $v_i$, $j\in\{0,\dots,n\}\setminus\{i-1,i,i+1\}$. 
Let $v_s$ and $v_t$ denote the matched vertices of $v_0$ and $v_1$, respectively.
Without loss of generality, we may assume that $v_0$ satisfies the statement of Lemma~\ref{lem:cub2},
so that $v_s\neq v_2$.
  
The following claim will be useful for constructing
an $L$-list incidence colouring of~$G$.

%%%%%%%%%%%%%%%%%%%%%%%%%%%%%%%%%%%%%%
\begin{claim}\label{claim:Ham}
There exist $a\in L(v_1,v_1v_t)$, $b\in L(v_s,v_sv_0)$,
$c\in L(v_{2},v_{2}v_{1})$, $d\in L(v_0,v_0v_s)$ and $e\in L(v_t,v_tv_{1})$,
with $a\neq c$, $a\neq e$ and $b\neq d$,
such that
$$|L(v_{0},v_{0}v_{1})\cap\{a,b\}|\le 1,\ \ \mbox{and}\ |L(v_1,v_1v_{0})\cap\{c,d,e\}|\le 1.$$
\end{claim}

\begin{proof}
We first deal with the incidence $(v_1,v_1v_{0})$ and set the values of $c$, $d$ and $e$
(see Figure~\ref{fig:Ham}).
Let $C=L(v_{2},v_{2}v_{1})$, $D=L(v_0,v_0v_s)$ and $E=L(v_t,v_tv_{1})$.
If $C\cap D\cap E\neq\emptyset$, then
we set $c=d=e=\gamma$ for some $\gamma\in C\cap D\cap E$,
so that $|L(v_1,v_1v_{0})\cap\{c,d,e\}|\le 1$.
Otherwise, we consider two cases.

\begin{enumerate}
\item If $C$, $D$ and $E$ are pairwise disjoint, then at least two of them
are distinct from $L(v_1,v_1v_{0})$, so that we can choose $c$, $d$ and $e$
in such a way that $|L(v_1,v_1v_{0})\cap\{c,d,e\}|\le 1$.
\item
Suppose now that $C\cap D\neq\emptyset$
(the cases $C\cap E\neq\emptyset$ and $D\cap E\neq\emptyset$ are similar).
We first set $c=d=\gamma$ for some $\gamma\in C\cap D$.
If $\gamma\in L(v_1,v_1v_{0})$, then there exists $\varepsilon\in E\setminus L(v_1,v_1v_{0})$
(since $(C\cap D)\cap E=\emptyset$) and we set $e=\varepsilon$,
so that $|L(v_1,v_1v_{0})\cap\{c,d,e\}|\le 1$.
If $\gamma\notin L(v_1,v_1v_{0})$, then we set $e=\varepsilon$ for any $\varepsilon\in E$
and we also get $|L(v_1,v_1v_{0})\cap\{c,d,e\}|\le 1$.
\end{enumerate}

We now deal with the incidence $(v_0,v_0v_{1})$ and set the values of $a$ and $b$.
Let $L'(v_1,v_1v_t)=L(v_1,v_1v_t)\setminus \{e,c\}$ and $L'(v_s,v_sv_0)=L(v_s,v_sv_0)\setminus\{d\}$.
If $L'(v_1,v_1v_t)\cap L'(v_s,v_sv_0)\neq \emptyset$, then we set $a=b=\alpha$ for some $\alpha\in L'(v_1,v_1v_t)\cap L'(v_s,v_sv_0)$, so that $|L(v_{0},v_{0}v_{1})\cap\{a,b\}|\le 1$.
Otherwise, as $|L(v,vu)|=6$ for every incidence $(v,vu)$ of $G$, 
which implies $|L'(v_1,v_1v_t)|\geq 4$ and $|L'(v_s,v_sv_0)|\geq 5$,
we get $|L'(v_1,v_1v_t)\cup L'(v_s,v_sv_0)|\geq 9$.
Therefore, either there exists some colour $\alpha\in L'(v_s,v_sv_0)\setminus L(v_0,v_0v_{1})$,
in which case we set $b=\alpha$, so that $|L(v_{0},v_{0}v_{1})\cap\{a,b\}|\le 1$ for any value of $a$,
or there exists some colour $\alpha\in L'(v_1,v_1v_t)\setminus L(v_0,v_0v_{1})$,
in which case we set $a=\alpha$, so that $|L(v_{0},v_{0}v_{1})\cap\{a,b\}|\le 1$ for any value of $b$.
This completes the proof of Claim~\ref{claim:Ham}.
\end{proof}
%%%%%%%%%%%%%%%%%%%%

%%%%%%%%%%%%%%%%%%%%%%%%%%%%%%%%%%%%%%%%%%%%%%%%%%%%%%%%%%%%%%
\begin{figure}
\begin{center}
\begin{tikzpicture}[domain=0:10,x=0.8cm,y=0.8cm]
%sommets
\SOMMET{-1,0} \SOMMET{2,0} \SOMMET{5,0} \SOMMET{8,0} 
\SOMMET{2,3} \SOMMET{5,3} 
%arêtes
\draw[very thick,dotted] (-3,0) -- (-2,0);
\draw[very thick] (-2,0) -- (9,0);
\draw[very thick,dotted] (9,0) -- (10,0);
\draw[very thick] (2,0) -- (2,3);
\draw[very thick] (5,0) -- (5,3);
%labels
\node[below] at (-1,0) {$v_{n-1}$};
\node[below] at (2,0) {$v_{0}$};
\node[below] at (5,0) {$v_{1}$};
\node[below] at (8,0) {$v_{2}$};
\node[above] at (2,3) {$v_s$};
\node[above] at (5,3) {$v_t$};
%couleurs
\node[above] at (7.3,0) {$c$};
\node[left] at (2,0.7) {$d$};
\node[left] at (2,2.3) {$b$};
\node[right] at (5,0.7) {$a$};
\node[right] at (5,2.3) {$e$};

\end{tikzpicture}
\caption{Configuration for the proof of Claim~\ref{claim:Ham}.}
\label{fig:Ham}
\end{center}
\end{figure}
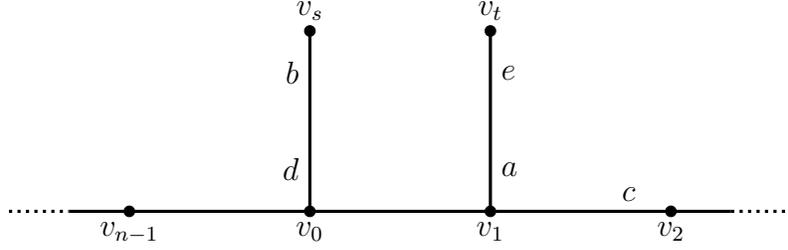
%%%%%%%%%%%%%%%%%%%%%%%%%%%%%%%%%%%%%%%%%%%%%%%%%%%%%%%%%%%%%%%%%%%%%%%%

We now construct an $L$-list incidence colouring $\sigma$ of~$G$ in three steps.

\begin{enumerate}
\item
We first set
$\sigma(v_1,v_1v_t)=a$, $\sigma(v_s,v_sv_0)=b$,
$\sigma(v_{2},v_{2}v_{1})=c$, $\sigma(v_0,v_0v_s)=d$ and $\sigma(v_t,v_tv_{1})=e$,
where $a$, $b$, $c$, $d$ and $e$ are the values determined in the proof
of Claim~\ref{claim:Ham}.

\item
We colour all the uncoloured incidences of the perfect matching $F=E(G)\setminus E(C_G)$.
This can be done since every such incidence has at most two already coloured
adjacent incidences (indeed, only the lastly coloured incidence of the edge $v_2v'_2$,
where $v'_2$ is the antipodal vertex of $v_2$, will have two already coloured
adjacent incidences).

\item
 We finally colour all the uncoloured incidences of $C_G$
(the only incidence of $C_G$ already coloured is $(v_2,v_2v_1)$) as follows.

\begin{itemize}
\item We first colour the incidence $(v_1,v_1v_2)$, which has
four already
coloured adjacent incidences.

\item We then cyclically colour the incidences of $C_G$ from $(v_2,v_2v_3)$
to $(v_{0},v_0v_{n-1})$.
%$(v_{n-1},v_{n-1}v_{0})$.
This can be done since, doing so, every such incidence has four or five already
coloured adjacent incidences.

\item By Claim~\ref{claim:Ham}, the incidence $(v_{0},v_{0}v_1)$ has at most
five forbidden colours and can thus be coloured.
Similarly, thanks to Claim~\ref{claim:Ham},
the incidence $(v_1,v_1v_0)$ has at most
five forbidden colours and can thus be coloured.
\end{itemize}
\end{enumerate}

This completes the proof of Theorem~\ref{th:HamiltonianCubic}.
\end{proof}

By Observation~\ref{obs:subgraph}, we get the following
corollary of Theorem~\ref{th:HamiltonianCubic}.

\begin{corollary}
If $G$ is a Hamiltonian graph with maximum degree~3, then $\ch_i(G)\leq 6$.
\label{cor:HamiltonianSubcubic}
\end{corollary}

%%%%%%%%%%%%%%%%%%%%%%%%%%%%%%%%%%%%%%%%%%%%%%%%%%%%%%%%%%%%%%%%%%%%%%%%%%%%%%%%%%%%%%%%%%%%%%%%%%%%%%%%%%%%%%%%%%%%%%%%%%%%%
%%%%%%%%%%%%%%%%%%%%%%%%%%%%%%%%%%%%%%%%%%%%%%%%%%%%%%%%%%%%%%%%%%%%%%%%%%%%%%%%%%%%%%%%%%%%%%%%%%%%%%%%%%%%%%%%%%%%%%%%%%%%%
%%%%%%%%%%%%%%%%%%%%%%%%%%%%%%%%% DISCUSSION

\section{Discussion}\label{sec:discussion}

In this paper, we have introduced and studied the list version of incidence colouring.
We determined the exact value of -- or upper bounds on -- the incidence choice number of 
several classes of graphs, namely square grids, Halin graphs, generalized coronae of cycles,
cactuses and Hamiltonian cubic graphs.
Following the work presented here, we propose the following problems:

\begin{enumerate}
\item Is it true that $\ch_i(G_{m,n})=6$ for every square grid $G_{m,n}$ with $m\geq n\geq 3$?
\item What is the best possible upper bound on the list incidence chromatic number of Halin graphs
with maximum degree~3, 4 or 5? (Theorem~\ref{th:Halin} gives the exact bound only for Halin graphs
with maximum degree $k\ge 6$.)
\item What is the best possible upper bound on the list incidence chromatic number of cactuses
with maximum degree~6? 
With maximum degree~5 and containing at most one maximal cycle? 
With maximum degree~4 and containing a maximal cycle?
(Theorem~\ref{th:cactus} gives the exact bound for all other cases.)
\item What is the best possible upper bound on the list incidence chromatic number of graphs
with bounded maximum degree? In particular, what about graphs with maximum degree~3?
(By Proposition~\ref{prop:bounds-2}, we know that this bound is at most $3k-2$ for graphs with maximum
degree $k\ge 2$, and thus at most~7 for graphs with maximum degree~3.)
\item What is the value of $\ch_i(K_n)$?
(By Proposition~\ref{prop:bounds-2}, we know that this value is at most $3n-5$.)
\item Which classes of graphs satisfy the incidence version of the list colouring conjecture, that is,
for which graphs $G$ do we have $\ch_i(G)=\chi_i(G)$?
(By Proposition~\ref{prop:bounds} and Theorem~\ref{th:degenerated}, we know for instance that this equality holds for every tree.)
\end{enumerate}

\end{document}